\documentclass[12pt,draftclsnofoot,onecolumn]{IEEEtran}

\usepackage{url}
\usepackage{cite}
\usepackage{soul}
\usepackage{tikz}
\usepackage[normalem]{ulem}
\usepackage{array}
\usepackage{color}
\usepackage{float}
\usepackage{amsthm}
\usepackage{lipsum}
\usepackage{amsmath}
\usepackage{amssymb}
\usepackage{grffile}
\usepackage{optidef}
\usepackage{enumitem}
\usepackage{graphicx}
\usepackage{pgfplots}
\usepackage{stmaryrd}
\usepackage{thmtools}
\usepackage{clipboard}
\usepackage{tikzscale}
\usepackage{algorithm}
\usepackage{clipboard}
\usepackage{mathtools}
\usepackage{microtype}
\usepackage{subcaption}
\usepackage{algorithmic}
\usepackage{thm-restate}
\usepackage[font=small,labelfont=bf]{caption}

\allowdisplaybreaks

\theoremstyle{general} \newtheorem{theorem}{Theorem}
\theoremstyle{general} 
\theoremstyle{general} 
\theoremstyle{general} 
\theoremstyle{general} 
\theoremstyle{general} 
\theoremstyle{remark} \newtheorem{remark}{Remark}

\newlength{\mycolwidth}
\begin{document}
	\title{Uplink and Downlink MIMO-NOMA with Simultaneous Triangularization}
	\author{\IEEEauthorblockN{Aravindh Krishnamoorthy\rlap{\textsuperscript{\IEEEauthorrefmark{2}\IEEEauthorrefmark{1}}}\,\,\,  and Robert Schober\rlap{\textsuperscript{\IEEEauthorrefmark{2}}}\\
	\IEEEauthorblockA{\small \IEEEauthorrefmark{2}Friedrich-Alexander-Universit\"{a}t Erlangen-N\"{u}rnberg, \IEEEauthorrefmark{1}Fraunhofer Institute for Integrated Circuits (IIS) Erlangen}}
	\thanks{This paper was presented in part at the IEEE Global Commun. Conf. (Globecom) 2019 \cite{Krishnamoorthy2019a} and accepted for presentation in part at the IEEE Wireless Commun. and Netw. Conf. (WCNC) 2021 \cite{Krishnamoorthy2020a}. Computer programs for the most important results in this paper can be downloaded from \protect\url{https://gitlab.com/aravindh.krishnamoorthy/mimo-noma}}}
	\maketitle
	
	\begin{abstract}
	In this paper, we consider the uplink and downlink precoder design for two-user power-domain multiple-input multiple-output (MIMO) non-orthogonal multiple access (NOMA) systems. We propose novel uplink and downlink precoding  and detection schemes that lower the decoding complexity at the receiver by decomposing the MIMO-NOMA channels of the users into multiple single-input single-output (SISO)-NOMA channels via simultaneous triangularization (ST) of the MIMO channels of the users and low-complexity self-interference cancellation at the receivers. The proposed ST MIMO-NOMA schemes avoid channel inversion at transmitter and receiver and take advantage of the null spaces of the MIMO channels of the users, which is beneficial for the ergodic achievable rate performance. We characterize the maximum ergodic achievable rate regions of the proposed uplink and downlink ST MIMO-NOMA schemes, and compare them with respective upper bounds, baseline MIMO-NOMA schemes, and orthogonal multiple access (OMA). Our results illustrate that the proposed schemes significantly outperform the considered baseline MIMO-NOMA schemes and OMA, and have a small gap to the respective upper bounds for most channel conditions and user rates. Moreover, we show that a hybrid scheme, which performs time sharing between the proposed uplink and downlink ST MIMO-NOMA and single-user MIMO, can improve performance even further.
	\end{abstract}
	
	\section{Introduction}
	Non-orthogonal multiple access (NOMA) has the potential to improve the spectral efficiency and data rate of 5th generation (5G) and beyond mobile communication systems \cite{Saito2013}. Specifically, power-domain NOMA, which utilizes superposition coding at the transmitter and successive interference cancellation (SIC) at the receiver, is of interest owing to its compatibility with the 4th generation (4G) communication systems. In the NOMA literature, research has mostly focused on single-input single-output (SISO)-NOMA so far \cite{Dai2018}, \cite{Ding2017}. However, recently, the extension of NOMA to multiple-input multiple-output (MIMO) systems has garnered interest owing to the potential performance gains compared to traditional MIMO orthogonal multiple access (OMA) schemes \cite{Zeng2017}. Unfortunately, capacity achieving MIMO-NOMA schemes are too complex for practical implementation \cite{Dai2018}. Therefore, several low-complexity MIMO-NOMA schemes have been proposed.
	
	For uplink transmission, a MIMO-NOMA scheme based on generalized singular value decomposition (GSVD) was proposed in \cite{Ma2016} for the special case where the numbers of antennas at the base station (BS) and the user equipments (UEs) are equal. Moreover, uplink MIMO-NOMA schemes based on zero forcing (ZF) and minimum mean squared error (MMSE) decoding were reported in \cite{Endo2012}. An iterative linear minimum mean-square estimation (LMMSE) based decoding scheme using parallel interference cancellation (PIC) for MIMO-NOMA was presented in \cite{Liu2019a}. Signal alignment (SA) based MIMO-NOMA, proposed in \cite{Ding2016}, uses multiple antennas at the users for uplink beamforming in order to avoid interference at the BS. SA reduces the decoding complexity at the BS by decomposing the MIMO-NOMA channel into multiple SISO-NOMA channels \cite{Saito2013}. However, SA necessitates multiple antennas at the users, which may not be feasible for all mobile user devices due to size constraints. Furthermore, despite the use of multiple antennas at the users, in SA based MIMO-NOMA, a user can only transmit a single spatial stream to the BS as the remaining degrees of freedom (DoFs) are used for interference cancellation. Hence, the development of uplink MIMO-NOMA schemes that allow each user to transmit multiple spatial streams, while maintaining a low decoding complexity, is of high practical interest.
		
	On the other hand, for downlink transmission, several precoding schemes for MIMO-NOMA have been reported \cite{Chen2016,Chen2016a,Ding2016,Ali2017,Zeng2017a,Chen2017a,Ding2016b,Choi2016,Rezaei2020,Xiao2019a,Tong2019,Morales-Cespedes2019}. Furthermore, power allocation for MIMO-NOMA systems was investigated in \cite{Wang2019,Xiao2019,Zhang2020,Panda2020}. Moreover, in order to reduce the decoding complexity at the users, precoder designs that simultaneously diagonalize the users' MIMO channels were reported in \cite{Krishnamoorthy2020},\cite{Chen2019}. Optimal power allocation for the precoder in \cite{Chen2019} was studied in \cite{Hanif2019}. The simultaneous diagonalization (SD)\footnote{SD decomposes the MIMO channel of the users into diagonal matrices utilizing a linear precoder and a per-user detection matrix.} based precoding and detection schemes in \cite{Krishnamoorthy2020},\cite{Chen2019} facilitate low-complexity decoding at the users by decomposing the MIMO-NOMA channels of the users into multiple parallel single-input single-output (SISO)-NOMA channels. Furthermore, these schemes exploit the available null spaces of the MIMO channels of the users for enabling inter-user-interference free communication, thereby enhancing the ergodic rate performance. However, the precoding schemes in \cite{Krishnamoorthy2020},\cite{Chen2019} achieve SD by inverting the MIMO channels of the users, which limits their performance.
	
	Hence, in this paper, we propose uplink and downlink MIMO-NOMA precoding and detection schemes based on simultaneous triangularization (ST) of the MIMO channels of the users, enabling low-complexity decoding at the receiver. Although the proposed precoding and detection schemes avoid inversion of the MIMO channels of the users, they can still take advantage of the null spaces of the MIMO channels of the users to achieve inter-user-interference free communication.
	
	This paper builds upon the conference versions in \cite{Krishnamoorthy2019a} and \cite{Krishnamoorthy2020a}. The uplink ST MIMO-NOMA scheme in \cite{Krishnamoorthy2019a} was limited to the case where the number of BS antennas is larger than the numbers of user antennas and was analyzed for equal power allocation. In this paper, we extend the scheme to all possible antenna configurations and optimal power allocation. On the other hand, suboptimal power allocation was considered for the downlink ST MIMO-NOMA scheme in \cite{Krishnamoorthy2020a} leading to a lower bound on the achievable rate region. In this paper, we determine the maximum  achievable rate region of downlink ST MIMO-NOMA based on a corresponding broadcast channel (BC) to multiple access channel (MAC) transformation analogous to that in \cite{Jindal2004}. Furthermore, we evaluate the performance of the proposed ST MIMO-NOMA schemes for a broader set of scenarios compared to \cite{Krishnamoorthy2019a} and \cite{Krishnamoorthy2020a}, respectively. The main contributions of this paper can be summarized as follows.
	
	\begin{itemize}
		\item Exploiting the QR decomposition, we develop uplink and downlink ST MIMO-NOMA precoding and detection schemes and corresponding low-complexity decoding schemes which decompose the downlink MIMO-NOMA channel into multiple parallel SISO-NOMA channels, assuming low-complexity self-interference cancellation at the users.
		\item For uplink transmission, we characterize the maximum achievable rate region based on convex optimization, and for downlink transmission, we exploit a BC-MAC transformation \cite{Jindal2004} and polyblock outer approximation \cite{Tuy2000} to obtain the corresponding maximum achievable rate region.
		\item Lastly, based on the obtained maximum achievable rate regions, we show that, for both uplink and downlink transmission, the proposed ST MIMO-NOMA schemes outperform baseline MIMO-NOMA schemes and MIMO-OMA for most channel conditions and user rates.
	\end{itemize}
	
	The remainder of this paper is organized as follows. We establish the considered uplink and downlink system models in Section \ref{sec:prelim}. In Section \ref{sec:propuplink}, we present the proposed uplink ST MIMO-NOMA precoding, detection, and decoding schemes and expressions for the corresponding achievable user rates. The proposed downlink ST MIMO-NOMA scheme is presented in Section \ref{sec:propdownlink}. In Section \ref{sec:opa}, we characterize the maximum achievable rate regions for uplink and downlink ST MIMO-NOMA. Simulation results are presented in Section \ref{sec:sim}, and conclusions are provided in Section \ref{sec:con}.
	
	\emph{Notation:} Boldface capital letters $\boldsymbol{X}$ and lower case letters $\boldsymbol{x}$ denote matrices and vectors, respectively. $\boldsymbol{X}^\mathrm{T}$, $\boldsymbol{X}^\mathrm{H}$,  $\boldsymbol{X}^+$, $\mathrm{tr}\left(\boldsymbol{X}\right)$, and $\mathrm{det}\left(\boldsymbol{X}\right)$ denote the transpose, Hermitian transpose, Moore-Penrose pseudoinverse, trace, and determinant of matrix $\boldsymbol{X}$, respectively. Furthermore, $\mathrm{col}\big(\boldsymbol{X}\big)$ and $\mathrm{null}\left(\boldsymbol{X}\right)$ denote the column space and null space of matrix $\boldsymbol{X},$ respectively. $\mathbb{C}^{m\times n}$ and $\mathbb{R}^{m\times n}$ denote the sets of all $m\times n$ matrices with complex-valued and real-valued entries, respectively. The $(i,j)$-th entry of matrix $\boldsymbol{X}$ is denoted by $[\boldsymbol{X}]_{ij}$ and the $i$-th entry of vector $\boldsymbol{x}$ is denoted by $[\boldsymbol{x}]_i.$ $\boldsymbol{I}_n$ denotes the $n\times n$ identity matrix, and $\boldsymbol{0}$ denotes the all zero matrix of appropriate dimension. The circularly symmetric complex Gaussian (CSCG) distribution with mean vector $\boldsymbol{\mu}$ and covariance matrix $\boldsymbol{\Sigma}$ is denoted by $\mathcal{CN}(\boldsymbol{\mu},\boldsymbol{\Sigma})$; $\sim$ stands for ``distributed as''. $\mathrm{E}[\cdot]$ denotes statistical expectation.
	
	\section{Preliminaries}
	\label{sec:prelim}
	In this section, we present the considered two-user power-domain MIMO-NOMA system model. We consider a communication system with a BS employing $N$ antennas and two users\footnote{For both uplink and downlink, we restrict the number of paired users to two for problem tractability. Extending the proposed ST MIMO-NOMA schemes to more than two users while retaining its desirable properties that enable inter-user-interference free communication seems very challenging and is beyond the scope of this paper. For $K>2$ users, a hybrid approach, such as in \cite[Section V-B]{Chen2019}, can be employed where the users are divided into groups of two users and each group is allocated orthogonal resources. Within each two-user group, the proposed ST MIMO-NOMA schemes can be applied.} whose UEs are equipped with $M_k, k=1,2,$ antennas. Furthermore, we assume that the first user is located farther away from the BS compared to the second user, thereby experiencing a higher path loss\footnote{Pairing users experiencing different channel conditions is crucial for exploiting the benefits of NOMA \cite{Saito2013}. The user labels `first' and `second' can be adjusted so that the user located farther away from the BS is always labeled as the `first' user.}.
	\subsection{Channel Model}
	\label{sec:cm}
	The downlink MIMO channel between the $k$-th user, $k=1,2,$ and the BS is modeled as
	\begin{equation}
	\frac{1}{\sqrt{\mathstrut \Pi_k}} \boldsymbol{H}_k,
	\end{equation}
	where the elements of matrix $\boldsymbol{H}_k \in \mathbb{C}^{M_k\times N}, k=1,2,$ model small-scale fading effects. Furthermore, $\Pi_k > 0, k=1,2,$ models the path loss between the $k$-th user and the BS, where $\Pi_1 > \Pi_2$. The uplink MIMO channel between the $k$-th user, $k=1,2,$ and the BS is modeled correspondingly as
	\begin{equation}
	\frac{1}{\sqrt{\mathstrut \Pi_k}} \boldsymbol{H}_k^\mathrm{H}.
	\end{equation}
	Furthermore, perfect knowledge of both MIMO channel matrices, $\boldsymbol{H}_1$ and $\boldsymbol{H}_2,$ is assumed at the BS, whereas perfect knowledge of their respective MIMO channel matrices is assumed at the users\footnote{Perfect channel knowledge is assumed to obtain performance upper bounds for the proposed ST MIMO-NOMA schemes. In practice, the MIMO channel matrices can be acquired at the BS and the users based on uplink and downlink pilots, respectively, analogous to the case of single-user MIMO systems, see e.g. \cite{Minn2006}.}.
	
	\subsection{Uplink System Model}
	\label{sec:usm}
	Let $L_k = \mathrm{min}\left\{M_k, N\right\}$ denote the symbol vector length of the $k$-th user, $k=1,2,$ and let $\boldsymbol{s}_1^{\mathrm{U}} = [s_{1,1}^{\mathrm{U}},\dots,s_{1,L_1}^{\mathrm{U}}]^\mathrm{T} \in \mathbb{C}^{L_1\times 1}$ and $\boldsymbol{s}_2^{\mathrm{U}} = [s_{2,1}^{\mathrm{U}},\dots,s_{2,L_2}^{\mathrm{U}}]^\mathrm{T} \in \mathbb{C}^{L_2\times 1}$ denote the transmit symbol vectors of the first and the second user, respectively. Here, we assume that the $s_{k,l}^{\mathrm{U}} \sim \mathcal{CN}(0,1),k=1,2,l=1,\dots,L_k,$ are i.i.d\footnote{In this work, we assume ideal Gaussian signaling for evaluation of the achievable user rates for ST MIMO-NOMA. Modulation and coding schemes that can closely approach the performance of ideal Gaussian signaling are known and can be utilized in practical implementations, see e.g. \cite{Liu2020}.}.
	
	User $k$ precodes its transmit symbol vector using a linear precoder matrix $\boldsymbol{P}_k^{\mathrm{U}} \in \mathbb{C}^{M_k\times L_k},k=1,2.$ Both users transmit their precoded symbol vectors to the BS over the same resource. The transmit power of user $k,$ denoted by $P_k^{\mathrm{U}},k=1,2,$ is given by
	\begin{align}
		P_k^{\mathrm{U}} = \mathrm{tr}\left(\boldsymbol{P}_k^{\mathrm{U}} (\boldsymbol{P}_k^{\mathrm{U}})^\mathrm{H}\right).
	\end{align}
	The received signal at the BS, $\boldsymbol{y}'^{\mathrm{U}} \in \mathbb{C}^{N\times 1},$ is given by
	\begin{equation}
		\boldsymbol{y}'^{\mathrm{U}} = \frac{1}{\sqrt{\mathstrut \Pi_1}} \boldsymbol{H}_1^\mathrm{H} \boldsymbol{P}_1^{\mathrm{U}} \boldsymbol{s}_1^{\mathrm{U}} + \frac{1}{\sqrt{\mathstrut \Pi_2}}\boldsymbol{H}_2^\mathrm{H} \boldsymbol{P}_2^{\mathrm{U}} \boldsymbol{s}_2 ^{\mathrm{U}}+ \boldsymbol{n}'^{\mathrm{U}},
	\end{equation}
	where $\boldsymbol{n}'^{\mathrm{U}} \sim \mathcal{CN}(\boldsymbol{0}, \sigma^2\boldsymbol{I}_N)$ denotes the additive white Gaussian noise (AWGN) vector at the BS. Prior to decoding, the received signal at the BS is processed using a unitary detection matrix $\boldsymbol{Q}^{\mathrm{U}} \in \mathcal{C}^{N\times N}$ to obtain
	\begin{equation}
		\boldsymbol{y}^{\mathrm{U}} = \boldsymbol{Q}^{\mathrm{U}}\boldsymbol{y}'^{\mathrm{U}} = \frac{1}{\sqrt{\mathstrut \Pi_1}}\boldsymbol{Q}^{\mathrm{U}}\boldsymbol{H}_1^\mathrm{H} \boldsymbol{P}_1^{\mathrm{U}} \boldsymbol{s}_1^{\mathrm{U}} + \frac{1}{\sqrt{\mathstrut \Pi_2}}\boldsymbol{Q}^{\mathrm{U}}\boldsymbol{H}_2^\mathrm{H} \boldsymbol{P}_2^{\mathrm{U}} \boldsymbol{s}_2^{\mathrm{U}} + \boldsymbol{n}^{\mathrm{U}}, \label{eqn:y}
	\end{equation}
	where $\boldsymbol{n}^{\mathrm{U}} = \boldsymbol{Q}^{\mathrm{U}}\boldsymbol{n}'^{\mathrm{U}} \sim \mathcal{CN}(\boldsymbol{0}, \sigma^2\boldsymbol{I}_N).$ $\boldsymbol{y}^{\mathrm{U}}$ is subsequently used for decoding.
	
	\subsection{Downlink System Model}
	\label{sec:dsm}
	Let $L = \mathrm{min}\left\{M_1+M_2, N\right\}$ denote the symbol vector length, and let $\boldsymbol{s}_1^{\mathrm{D}} = [s_{1,1}^{\mathrm{D}},\dots,s_{1,L}^{\mathrm{D}}]^\mathrm{T} \in \mathbb{C}^{L\times 1}$ and $\boldsymbol{s}_2^{\mathrm{D}} = [s_{2,1}^{\mathrm{D}},\dots,s_{2,L}^{\mathrm{D}}]^\mathrm{T} \in \mathbb{C}^{L\times 1}$ denote the symbol vectors intended for the first and the second users, respectively. We assume that the $s_{k,l}^{\mathrm{D}} \sim \mathcal{CN}(0,1), k=1,2, l=1,\dots,L,$ are i.i.d. for all $k,l,$ as in the uplink case. We construct the downlink MIMO-NOMA symbol vector $\boldsymbol{s}^{\mathrm{D}} = [s_1^{\mathrm{D}}, \dots, s_L^{\mathrm{D}}]^\mathrm{T}$ as follows
	\begin{align}
	\boldsymbol{s}^{\mathrm{D}} = \mathrm{diag}\left(\sqrt{p_{1,1}^{\mathrm{D}}},\dots,\sqrt{p_{1,L}^{\mathrm{D}}}\right)\boldsymbol{s}_1^{\mathrm{D}} + \mathrm{diag}\left(\sqrt{p_{2,1}^{\mathrm{D}}},\dots,\sqrt{p_{2,L}^{\mathrm{D}}}\right)\boldsymbol{s}_2^{\mathrm{D}}, \label{eqn:s}
	\end{align}
	where $p_{k,l}^{\mathrm{D}} \geq 0, k=1,2, l=1,\dots,L,$ is the transmit power allocated to the $l$-th symbol of user $k.$ The MIMO-NOMA symbol vector is precoded using a linear precoder matrix $\boldsymbol{P}^{\mathrm{D}} \in \mathbb{C}^{N\times L}$ resulting in transmit signal $\boldsymbol{x}^{\mathrm{D}} = \boldsymbol{P}^{\mathrm{D}} \boldsymbol{s}^{\mathrm{D}}.$ The corresponding transmit power is given by
	\begin{align}
		P_\mathrm{T}^{\mathrm{D}} = \mathrm{tr}\left(\boldsymbol{P}^{\mathrm{D}} \mathrm{diag}\left(p_{1,1}^{\mathrm{D}} + p_{2,1}^{\mathrm{D}}, \dots, p_{1,L}^{\mathrm{D}} + p_{2,L}^{\mathrm{D}}\right) (\boldsymbol{P}^{\mathrm{D}})^\mathrm{H}\right). \label{eqn:eppleq1}
	\end{align}
	
	At user $k,$ $k=1,2,$ the received signal, $\boldsymbol{y}_k'^{\mathrm{D}} \in \mathbb{C}^{M_k\times 1},$ is given by
	\begin{align}
	\boldsymbol{y}_k'^{\mathrm{D}} &= \frac{1}{\sqrt{\Pi_k}}\boldsymbol{H}_k \boldsymbol{x}^{\mathrm{D}} + \boldsymbol{n}_k'^{\mathrm{D}} = \frac{1}{\sqrt{\Pi_k}} \boldsymbol{H}_k \boldsymbol{P}^{\mathrm{D}} \boldsymbol{s}^{\mathrm{D}} + \boldsymbol{n}_k'^{\mathrm{D}},
	\end{align} 
	where $\boldsymbol{n}_k'^{\mathrm{D}} \sim \mathcal{CN}(\boldsymbol{0},\sigma^2 \boldsymbol{I}_{M_k})$ denotes the AWGN vector at user $k.$ Furthermore, at user $k,$ signal $\boldsymbol{y}_k'^{\mathrm{D}}$ is processed by a unitary detection matrix $\boldsymbol{Q}_k^{\mathrm{D}} \in \mathbb{C}^{M_k\times M_k}$ leading to
	\begin{equation}
	\boldsymbol{y}_k^{\mathrm{D}} = \boldsymbol{Q}_k^{\mathrm{D}} \boldsymbol{y}_k'^{\mathrm{D}} = \frac{1}{\sqrt{\Pi_k}} \boldsymbol{Q}_k^{\mathrm{D}}\boldsymbol{H}_k \boldsymbol{P}^{\mathrm{D}} \boldsymbol{s}^{\mathrm{D}} + \boldsymbol{n}_k^{\mathrm{D}}, \label{eqn:yk}
	\end{equation}
	where $\boldsymbol{n}_k^{\mathrm{D}} = \boldsymbol{Q}_k^{\mathrm{D}}\boldsymbol{n}_k'^{\mathrm{D}} \sim \mathcal{CN}(\boldsymbol{0},\sigma^2 \boldsymbol{I}_{M_k}).$ $\boldsymbol{y}_k^{\mathrm{D}}$ is subsequently used for decoding.
	
	\section{Proposed Uplink ST MIMO-NOMA Scheme}
	\label{sec:propuplink}
	In this section, first, we develop a matrix decomposition for ST of the uplink MIMO channels of the users. Next, we exploit this matrix decomposition to design the proposed uplink ST MIMO-NOMA precoding and decoding schemes, and provide the corresponding achievable rate expressions.
	
	\subsection{Simultaneous Triangularization for Uplink Transmission}
	\begin{theorem}
		\label{th:stu}
		Let $\boldsymbol{A}_1 \in \mathbb{C}^{p\times q_1}$ and $\boldsymbol{A}_2 \in \mathbb{C}^{p\times q_2}$ denote matrices with an equal number of rows and full row or column rank. Then, there exist unitary matrices $\boldsymbol{U} \in \mathbb{C}^{p\times p},$ $\boldsymbol{V}_1 \in \mathbb{C}^{q_1\times \mathrm{min}\left\{p,q_1\right\}},$ and $\boldsymbol{V}_2 \in \mathbb{C}^{q_2\times \mathrm{min}\left\{p,q_2\right\}}$ such that
		\begin{align}
			\boldsymbol{U}\boldsymbol{A}_1\boldsymbol{V}_1 &= \begin{bmatrix}\check{\boldsymbol{R}}_1 \\ \boldsymbol{0}\end{bmatrix}, \label{eqn:stu1}\\
			\boldsymbol{U}\boldsymbol{A}_2\boldsymbol{V}_2 &= \begin{bmatrix}\check{\boldsymbol{X}} \\ \check{\boldsymbol{R}}_2\end{bmatrix}, \label{eqn:stu2}
		\end{align}
		where $\check{\boldsymbol{R}}_1 \in \mathbb{C}^{\mathrm{min}\left\{p,q_1\right\}\times \mathrm{min}\left\{p,q_1\right\}}$ and $\check{\boldsymbol{R}}_2 \in \mathbb{C}^{\mathrm{min}\left\{p,q_2\right\}\times \mathrm{min}\left\{p,q_2\right\}}$ are upper-triangular matrices with real-valued entries along the main diagonals, and $\check{\boldsymbol{X}} \in \mathbb{C}^{(p-\mathrm{min}\left\{p,q_2\right\})\times \mathrm{min}\left\{p,q_2\right\}}$ is a full matrix. 
	\end{theorem}
	\begin{proof}
		Please refer Appendix \ref{app:stu}.
	\end{proof}
	
	In the following, we denote the simultaneous triangularization operation in Theorem \ref{th:stu} as
	\begin{equation}
	\triangle(\boldsymbol{A}_1, \boldsymbol{A}_2) \rightarrow (\boldsymbol{U}, \boldsymbol{V}_1, \boldsymbol{V}_2, \check{\boldsymbol{R}}_1, \check{\boldsymbol{R}}_2,\check{\boldsymbol{X}}).
	\end{equation}

	
	\subsection{Uplink Decoding Order}
	\label{sec:udo}
	From (\ref{eqn:y}), we observe that the user signals are superimposed at the BS. Hence, as described in detail later, successive interference cancellation (SIC) is employed at the BS. The achievable rates of the users depend on the decoding order. For decoding order (D-2-1), the symbols of the second user are decoded first, followed by SIC, and then the symbols of the first user are decoded. Similarly, (D-1-2) implies that the symbols of the first user are decoded first. Furthermore, for decoding orders (D-2-1) and (D-1-2), simultaneous triangularization operations $\triangle(\boldsymbol{H}_1^\mathrm{H}, \boldsymbol{H}_2^\mathrm{H})$ and $\triangle(\boldsymbol{H}_2^\mathrm{H}, \boldsymbol{H}_1^\mathrm{H})$ are utilized, respectively.

	\subsection{Proposed Uplink Pre-processing Scheme for (D-2-1)}
	\label{sec:uprecoding}
	As explained above, for (D-2-1), the second user's symbols are decoded first. Let $\bar{M}_1 = \mathrm{max}\left\{0,\mathrm{min}\left\{L_1,N-M_2\right\}\right\}, \bar{M}_2 = \mathrm{max}\left\{0,\mathrm{min}\left\{L_2,N-M_1\right\}\right\},$ and $M = N-\bar{M}_1-\bar{M}_2.$ Based on Theorem \ref{th:stu}, the precoding and detection matrices are chosen as follows\footnote{In practice, the precoding matrices can be computed at the BS based on uplink pilots and subsequently forwarded to the users.}.
	\begin{align}
		\triangle(\boldsymbol{H}_1^\mathrm{H},\boldsymbol{H}_2^\mathrm{H}) &\rightarrow (\boldsymbol{Q}^{\mathrm{U}},\boldsymbol{P}_1'^{\mathrm{U}}, \boldsymbol{P}_2'^{\mathrm{U}},\boldsymbol{R}_1^{\mathrm{U}},\boldsymbol{R}_2^{\mathrm{U}},\boldsymbol{X}^{\mathrm{U}}), \\
		\boldsymbol{P}_1^{\mathrm{U}} &= \boldsymbol{P}_1'^{\mathrm{U}}\underbrace{\mathrm{diag}\left(\sqrt{p_{1,1}^{\mathrm{U}}},\dots,\sqrt{p_{1,L_1}^{\mathrm{U}}}\right)}_{\coloneqq \boldsymbol{D}_1^{\mathrm{U}}}, \\
		\boldsymbol{P}_2^{\mathrm{U}} &= \boldsymbol{P}_2'^{\mathrm{U}}\underbrace{\mathrm{diag}\left(\sqrt{p_{2,1}^{\mathrm{U}}},\dots,\sqrt{p_{2,L_2}^{\mathrm{U}}}\right)}_{\coloneqq \boldsymbol{D}_2^{\mathrm{U}}},
	\end{align}
	where $\boldsymbol{Q}^{\mathrm{U}} \in \mathbb{C}^{N\times N}$ is unitary, $\boldsymbol{R}_1^{\mathrm{U}} \in \mathbb{C}^{L_1\times L_1}$ and $\boldsymbol{R}_2^{\mathrm{U}} \in \mathbb{C}^{L_2\times L_2}$ are upper-triangular matrices with real-valued entries along the main diagonals, and $\boldsymbol{X}^{\mathrm{U}} \in \mathbb{C}^{\bar{M}_1\times L_2}$ is a full matrix, and $p_{k,l}^{\mathrm{U}}, k=1,2,l=1,\dots,L_k,$ are power allocation coefficients such that
	\begin{align}
		\sum_{l=1}^{L_k} p_{k,l}^{\mathrm{U}} = P_k^{\mathrm{U}},
	\end{align}
	for $k=1,2.$

	Substituting the proposed precoding and detection matrices in (\ref{eqn:y}), the received signal at the BS is given by
	\begin{align}
		\boldsymbol{y}^{\mathrm{U}} &= \frac{1}{\sqrt{\Pi_1}}\begin{bmatrix}\boldsymbol{R}_1^{\mathrm{U}}\\ \boldsymbol{0}\end{bmatrix} \boldsymbol{D}_1^{\mathrm{U}}\boldsymbol{s}_1^{\mathrm{U}} + \frac{1}{\sqrt{\Pi_2}}\begin{bmatrix} \boldsymbol{X}^{\mathrm{U}} \\ \boldsymbol{R}_2^{\mathrm{U}}\end{bmatrix} \boldsymbol{D}_2^{\mathrm{U}}\boldsymbol{s}_2^{\mathrm{U}} + \boldsymbol{n}^{\mathrm{U}} \nonumber\\
		&=\frac{1}{\sqrt{\Pi_1}}\scalebox{0.7}{\mbox{\ensuremath{\displaystyle \left[\begin{array}{c}
		\underbrace{\begin{array}{>{$}w{c}{\mycolwidth}<{$}>{$}w{c}{\mycolwidth}<{$}>{$}w{c}{\mycolwidth}<{$}>{$}w{c}{\mycolwidth}<{$}}
		\rho^{(1)}_{1,1} & \rho^{(1)}_{1,2} & \cdots & \rho^{(1)}_{1,L_1} \\
		& \rho^{(1)}_{2,2} &  \rho^{(1)}_{2,3} & \vdots \\
		& & \ddots & \vdots \\
		& & & \rho^{(1)}_{L_1,L_1}		\end{array}}_{\boldsymbol{R}_1^{\mathrm{U}}}\\
		\begin{array}{c}
		\\
		\boldsymbol{0}\\
		\\
		\end{array}		
		\end{array}\right]}}} \boldsymbol{D}_1^{\mathrm{U}}\boldsymbol{s}_1^{\mathrm{U}} 
		+ \frac{1}{\sqrt{\Pi_2}}\scalebox{0.7}{\mbox{\ensuremath{\displaystyle \left[\begin{array}{c}
			\underbrace{\begin{array}{>{$}w{c}{\mycolwidth}<{$}>{$}w{c}{\mycolwidth}<{$}>{$}w{c}{\mycolwidth}<{$}>{$}w{c}{\mycolwidth}<{$}}
					x_{1,1}^{\mathrm{U}} & x_{1,2}^{\mathrm{U}} & \cdots & x_{1,L_2}^{\mathrm{U}} \\
					\vdots & \vdots & \ddots & \vdots \\
					x_{\bar{M}_1,1}^{\mathrm{U}} & x_{\bar{M}_1,2}^{\mathrm{U}} & \cdots & x_{\bar{M}_1,L_2}^{\mathrm{U}}\end{array}}_{\boldsymbol{X^{\mathrm{U}}}} \\
			\underbrace{\begin{array}{>{$}w{c}{\mycolwidth}<{$}>{$}w{c}{\mycolwidth}<{$}>{$}w{c}{\mycolwidth}<{$}>{$}w{c}{\mycolwidth}<{$}}
					\rho^{(2)}_{1,1} & \rho^{(2)}_{1,2} & \cdots & \rho^{(2)}_{1,L_2} \\
					& \rho^{(2)}_{2,2} &  \rho^{(2)}_{2,3} & \vdots \\
					& & \ddots & \vdots \\
					& & & \rho^{(2)}_{L_2,L_2}\end{array}}_{\boldsymbol{R}_2^{\mathrm{U}}}
	\end{array}\right]}}} \boldsymbol{D}_2^{\mathrm{U}}\boldsymbol{s}_2^{\mathrm{U}} + \boldsymbol{n}^{\mathrm{U}}. \label{eqn:yp}
	\end{align}
	From (\ref{eqn:yp}), we note that the last $\bar{M}_2$ elements of $\boldsymbol{y}^{\mathrm{U}}$ depend only on the transmit symbols of the second user. Furthermore, the last $M$ rows of $\boldsymbol{R}_1^{\mathrm{U}}$ and the first $M$ rows of $\boldsymbol{R}_2^{\mathrm{U}}$ overlap resulting in inter-user-interference for $[\boldsymbol{y}^{\mathrm{U}}]_l, l=\bar{M}_1+1,\dots,\bar{M}_1+M,$ as shown in Figure \ref{fig:stu}. Rewriting (\ref{eqn:yp}) as scalar equations, we have
	\begin{align}
	[\boldsymbol{y}^{\mathrm{U}}]_{l} &= \frac{1}{\sqrt{\Pi_2}}\rho^{(2)}_{l-\bar{M}_1,l-\bar{M}_1} \sqrt{p_{2,l-\bar{M}_1}^{\mathrm{U}}} s_{2,l-\bar{M}_1}^{\mathrm{U}} + \frac{1}{\sqrt{\Pi_2}}\sum_{l' = l-\bar{M}_1+1}^{L_2} \rho^{(2)}_{l-\bar{M}_1,l'} \sqrt{p_{2,l'}^{\mathrm{U}}} s_{2,l'}^{\mathrm{U}} + [\boldsymbol{n}^{\mathrm{U}}]_l, \label{eqn:y3}
	\end{align}
	for $l = \bar{M}_1+M+1,\dots,N,$
	\begin{align}
	[\boldsymbol{y}^{\mathrm{U}}]_{l} &= \frac{1}{\sqrt{\Pi_1}}\rho^{(1)}_{l,l} \sqrt{p_{1,l}^{\mathrm{U}}} s_{1,l}^{\mathrm{U}} + \frac{1}{\sqrt{\Pi_2}}\rho^{(2)}_{l-\bar{M}_1,l-\bar{M}_1} \sqrt{p_{2,l-\bar{M}_1}^{\mathrm{U}}} s_{2,l-\bar{M}_1}^{\mathrm{U}} + \frac{1}{\sqrt{\Pi_1}}\sum_{l'=l+1}^{L_1} \rho^{(1)}_{l,l'} \sqrt{p_{1,l'}^{\mathrm{U}}} s_{1,l'}^{\mathrm{U}}  \nonumber\\&\qquad+  \frac{1}{\sqrt{\Pi_2}}\sum_{l' = l-\bar{M}_1+1}^{L_2}\rho^{(2)}_{l-\bar{M}_1,l'} \sqrt{p_{2,l'}^{\mathrm{U}}} s_{2,l'}^{\mathrm{U}} + [\boldsymbol{n}^{\mathrm{U}}]_l, \label{eqn:y2}
	\end{align}
	for $l = \bar{M}_1+1,\dots,\bar{M}_1+M,$ and
	\begin{align}
	[\boldsymbol{y}^{\mathrm{U}}]_{l} &= \frac{1}{\sqrt{\Pi_1}}\rho^{(1)}_{l,l} \sqrt{p_{1,l}^{\mathrm{U}}} s_{1,l}^{\mathrm{U}} + \frac{1}{\sqrt{\Pi_1}}\sum_{l'=l+1}^{L_1} \rho^{(1)}_{l,l'} \sqrt{p_{1,l'}^{\mathrm{U}}} s_{1,l'}^{\mathrm{U}} + \frac{1}{\sqrt{\Pi_2}}\sum_{l' = 1}^{L_2} x_{l,l'}^{\mathrm{U}} \sqrt{p_{2,l'}^{\mathrm{U}}} s_{2,l'}^{\mathrm{U}} + [\boldsymbol{n}^{\mathrm{U}}]_l,\label{eqn:y1}
	\end{align}
	for $l = 1,\dots,\bar{M}_1.$
	
	Lastly, the rates of $s_{1,l}^{\mathrm{U}}, l=1,\dots,L_1,$ and $s_{2,l}^{\mathrm{U}}, l=1,\dots,L_2,$ are chosen such that they lie within the achievable rate region, provided later in Section \ref{sec:urates}.
	
	\begin{figure*}
		\centering
		\begin{minipage}{0.5\textwidth}
			\centering
			\includegraphics[width=\textwidth]{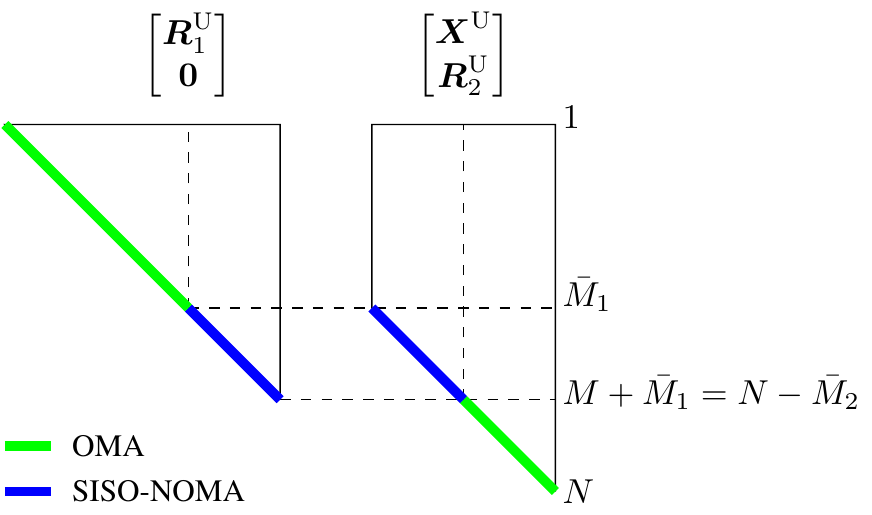}
			\caption{Simultaneous triangularization in uplink ST MIMO-NOMA for (D-2-1) decoding and $\bar{M}_1,\bar{M}_2,M > 0.$ The green and blue diagonal portions are decoded successively using conventional OMA and SISO-NOMA decoding, respectively.}
			\label{fig:stu}
		\end{minipage}\hspace{10pt}%
	\begin{minipage}{0.45\textwidth}
		\centering
		\captionof{table}{Complexity of precoder and detection matrix computation for the proposed uplink ST MIMO-NOMA and the considered MIMO-NOMA baseline schemes.}
		\label{tab:u}
		\begin{tabular}{|p{0.55\textwidth}|p{0.3\textwidth}|}
			\hline
			MIMO-NOMA Scheme & Complexity	 \\\hline\hline
			Proposed ST & $\mathcal{O}\mkern-\medmuskip
\left(4 N^3\right)$ \\\hline
			ZF and MMSE & $\mathcal{O}\mkern-\medmuskip
\left(5 N^3\right)$ \\\hline
			SVD & $\mathcal{O}\mkern-\medmuskip
\left(18 N^3\right)$ \\\hline
			GSVD & $\mathcal{O}\mkern-\medmuskip
\left(38.66 N^3\right)$ \\\hline
		\end{tabular}
	\end{minipage}
	\end{figure*}
	
	\subsection{Proposed Uplink Decoding Scheme for (D-2-1)}
	We observe that the elements of $\boldsymbol{y}^{\mathrm{U}},$ except for the last element, are subject to self-interference due to the triangularized channels. Furthermore, the first $M+\bar{M}_1$ elements are additionally subject to inter-user-interference due to the adopted superposition coding.
	
	Let $\hat{s}_{1,l}^{\mathrm{U}}, l=1,\dots,L_1,$ and $\hat{s}_{2,l}^{\mathrm{U}}, l=1,\dots,L_2,$ denote the detected symbols corresponding to transmitted symbols $s_{1,l}^{\mathrm{U}}$ and $s_{2,l}^{\mathrm{U}},$ respectively. As the rates of $s_{1,l}^{\mathrm{U}}, l=1,\dots,L_1,$ and $s_{2,l}^{\mathrm{U}}, l=1,\dots,L_2,$ are chosen such that they lie within the achievable rate region, $\hat{s}_{1,l}^{\mathrm{U}} = s_{1,l}^{\mathrm{U}}$ and $\hat{s}_{2,l}^{\mathrm{U}} = s_{2,l}^{\mathrm{U}}$ is assumed in the following. We adopt a decoding strategy where we decode element by element, in reverse order, beginning with the last element of $\boldsymbol{y}^{\mathrm{U}}, [\boldsymbol{y}^{\mathrm{U}}]_N,$ which is self-interference free.  For each element, we remove the self-interference of the previously decoded symbols. If $M > 0,$ for the middle $M$ elements, in addition to self-interference cancellation, we perform SIC to remove the inter-user-interference as in SISO-NOMA \cite{Saito2013}. The detailed decoding steps are as follows.
		
	\begin{enumerate}
		\item If $\bar{M}_2 > 0,$ as $[\boldsymbol{y}^{\mathrm{U}}]_{N}$ is both self- and inter-user-interference free, the corresponding symbol $s_{2,L_2}^{\mathrm{U}}$ is decoded directly.
		
		\item Next, if $\bar{M}_2 > 0,$ for $[\boldsymbol{y}^{\mathrm{U}}]_{l}, l=\bar{M}_1+M+1,\dots,N-1$ (in reverse order), assuming perfect decoding\footnote{The assumption of perfect decoding is justified as the rates corresponding to the decoded symbols are at or below the achievable rate $R_{k,l}$ provided in Section \ref{sec:urates}. In a practical implementation with practical modulation and coding schemes, decoding errors are unavoidable and may lead to error propagation. The impact of error propagation can be minimized by using powerful channel codes such as low density parity check (LDPC) codes.}, the self-interference from the previously decoded symbols of the second user is canceled to obtain a self-interference free signal
		\begin{align}
		[\hat{\boldsymbol{y}}^{\mathrm{U}}]_{l} &= [\boldsymbol{y}^{\mathrm{U}}]_{l} - \frac{1}{\sqrt{\Pi_2}}\sum_{l' = l-\bar{M}_1+1}^{L_2} \rho^{(2)}_{l-\bar{M}_1,l'} \sqrt{p_{2,l'}^{\mathrm{U}}} \hat{s}_{2,l'}^{\mathrm{U}} \nonumber\\&= \frac{1}{\sqrt{\Pi_2}}\rho^{(2)}_{l-\bar{M}_1,l-\bar{M}_1} \sqrt{p_{2,l-\bar{M}_1}^{\mathrm{U}}} s_{2,l-\bar{M}_1}^{\mathrm{U}}  + [\boldsymbol{n}^{\mathrm{U}}]_l, \label{eqn:yh3}
		\end{align}
		which is used for decoding $s_{2,M+1}^{\mathrm{U}},\dots,s_{2,L_2-1}^{\mathrm{U}},$ in reverse order.
		
		\item Next, if $M > 0,$ for $[\boldsymbol{y}^{\mathrm{U}}]_{l}, l = \bar{M}_1+1,\dots,\bar{M}_1+M$ (in reverse order), first the self-interference from the previously decoded symbols is canceled. The resulting signal
		\begin{align}
		[\hat{\boldsymbol{y}}^{\mathrm{U}}]_{l} &= [\boldsymbol{y}^{\mathrm{U}}]_{l} - \frac{1}{\sqrt{\Pi_1}}\sum_{l'=l+1}^{L_1} \rho^{(1)}_{l,l'} \sqrt{p_{1,l'}^{\mathrm{U}}} \hat{s}_{1,l'}^{\mathrm{U}} - \frac{1}{\sqrt{\Pi_2}}\sum_{l' = l-\bar{M}_1+1}^{L_2} \rho^{(2)}_{l-\bar{M}_1,l'} \sqrt{p_{2,l'}^{\mathrm{U}}} \hat{s}_{2,l'}^{\mathrm{U}} \nonumber\\
		&= \frac{1}{\sqrt{\Pi_1}}\rho^{(1)}_{l,l} \sqrt{p_{1,l}^{\mathrm{U}}} s_{1,l}^{\mathrm{U}} + \frac{1}{\sqrt{\Pi_2}}\rho^{(2)}_{l-\bar{M}_1,l-\bar{M}_1} \sqrt{p_{2,l-\bar{M}_1}^{\mathrm{U}}} s_{2,l-\bar{M}_1}^{\mathrm{U}} + [\boldsymbol{n}^{\mathrm{U}}]_l, \label{eqn:yh2}
		\end{align}
		contains residual inter-user-interference which is canceled as in SISO-NOMA \cite{Saito2013} where the second user's signal is decoded directly and the first user's signal is decoded after SIC. Hence, in this step, $s_{1,\bar{M}_1+1}^{\mathrm{U}},\dots,s_{1,L_1}^{\mathrm{U}}$ and $s_{2,1}^{\mathrm{U}},\dots,s_{2,M}^{\mathrm{U}}$ are decoded in reverse order.
		
		\item Lastly, if $\bar{M}_1 > 0,$ for elements $[\boldsymbol{y}^{\mathrm{U}}]_{l}, l=1,\dots,\bar{M}_1$ (in reverse order), self-interference and inter-user-interference from the previously decoded symbols is cancelled, resulting in the signal
		\begin{align}
		[\hat{\boldsymbol{y}}^{\mathrm{U}}]_{l} &= [\boldsymbol{y}^{\mathrm{U}}]_{l} - \frac{1}{\sqrt{\Pi_1}}\sum_{l'=l+1}^{L_1} \rho^{(1)}_{l,l'} \sqrt{p_{1,l'}^{\mathrm{U}}} \hat{s}_{1,l'}^{\mathrm{U}} - \frac{1}{\sqrt{\Pi_2}}\sum_{l' = 1}^{L_2} x_{l,l'}^{\mathrm{U}} \sqrt{p_{2,l'}^{\mathrm{U}}} \hat{s}_{2,l'}^{\mathrm{U}} \nonumber\\&= \frac{1}{\sqrt{\Pi_1}}\rho^{(1)}_{l,l} \sqrt{p_{1,l}^{\mathrm{U}}} s_{1,l}^{\mathrm{U}} + [\boldsymbol{n}^{\mathrm{U}}]_l, \label{eqn:yh1}
		\end{align}
		which is used for decoding $s_{1,1}^{\mathrm{U}},\dots,s_{1,\bar{M}_1}^{\mathrm{U}}$ in reverse order.
	\end{enumerate}
	
	\begin{remark}
	We note that after canceling the interference from the previously decoded symbols, $\bar{M}_1$ spatial streams of the first user and $\bar{M}_2$ spatial streams of the second user contain no inter-user-interference and can be decoded directly owing to the scheme's ability to exploit the null spaces of the MIMO channels of the users.
	\end{remark}
	
	\begin{remark}
		For decoding order (D-1-2), encoding and decoding are performed similarly as for decoding order (D-2-1), however, the roles of $\boldsymbol{s}_1$ and $\boldsymbol{s}_2$ and those of $\boldsymbol{H}_1^\mathrm{H}$ and $\boldsymbol{H}_2^\mathrm{H}$ are reversed.
	\end{remark}

	\subsection{Computational Complexity for Uplink Precoding}
	\label{sec:ucomp}
	To evaluate the complexity of computing the precoder and detection matrices, we consider the worst-case scenario $M_1=M_2=N.$ The proposed ST precoding, ZF- and MMSE-based precoding \cite{Chen2016, Chen2016a, Higuchi2015}, singular value decomposition (SVD)-based precoding\footnote{\scalebox{0.99}{SVD-based precoding matrices \cite{Gamal2011} are utilized to diagonalize the MIMO channels of the users, see Section \ref{sec:simuplink} for more details.}}, and GSVD-based precoding \cite{Ma2016} entail complexities of $\mathcal{O}\mkern-\medmuskip
\left(4N^3\right),$ $\mathcal{O}\mkern-\medmuskip
\left(5N^3\right)$ \cite{Krishnamoorthy2013}, $\mathcal{O}\mkern-\medmuskip
\left(18 N^3\right)$ \cite[Sec. 5.5]{Bai1992}, and $\mathcal{O}\mkern-\medmuskip
\left(38.66 N^3\right)$ \cite[Sec. 5.5]{Bai1992}, respectively, as summarized in Table \ref{tab:u}. For the proposed uplink scheme, the complexity of $\mathcal{O}\mkern-\medmuskip
\left(4 N^3\right)$ is incurred by two QR decompositions \cite[Alg. 5.2.5]{Golub2012}, see proof of Theorem \ref{th:stu}. Therefore, the proposed scheme entails an overall worst-case complexity of $\mathcal{O}\mkern-\medmuskip
\left(N^3\right),$ which is identical to those of ZF-, MMSE-, and GSVD-based precoding.
	
	In the following, we provide expressions for the achievable rates for (D-2-1) and (D-1-2) decoding.
	
	\subsection{Uplink Achievable Rates}
	\label{sec:urates}
	For (D-2-1) and for the second user, from (\ref{eqn:yh3}) and (\ref{eqn:yh2}), the achievable rate of symbols $s_{2,l}^{\mathrm{U}}, l=1,\dots,L_2,$ is given by\footnote{The achievable rates for (D-2-1) and (D-1-2) are indexed by $(1)$ and $(2),$ respectively.}
	\begin{equation}
	R_{2,l}^{{\mathrm{U}}, (1)} = \begin{cases}
	\log_2\left(1+\frac{\frac{1}{{\Pi_2}}p_{2,l}^{\mathrm{U}} (\rho^{(2)}_{l,l})^2}{\sigma^2 + \frac{1}{{\Pi_1}}p_{1,\bar{M}_1+l}^{\mathrm{U}} (\rho^{(1)}_{\bar{M}_1+l,\bar{M}_1+l})^2}\right) & \text{for $l=1,\dots,M$} \\
	\log_2\left(1+ \frac{1}{{\Pi_2}}\frac{p_{2,l}^{\mathrm{U}} (\rho^{(2)}_{l,l})^2}{\sigma^2}\right) & \text{for $l=M+1,\dots,L_2.$}
	\end{cases} \label{eqn:rsuf2}
	\end{equation}
	
	Furthermore, from (\ref{eqn:yh2}) and (\ref{eqn:yh1}), assuming successful SIC, the achievable rate of symbols $s_{1,l}^{\mathrm{U}}, l=1,\dots,L_1,$ is given by
	\begin{equation}
	\begin{array}{lr}
	R_{1,l}^{{\mathrm{U}}, (1)} = \log_2\left(1+\frac{1}{{\Pi_1}}\frac{p_{1,l}^{\mathrm{U}} (\rho^{(1)}_{l,l})^2}{\sigma^2}\right) & \text{for $l=1,\dots,L_1.$}
	\end{array} \label{eqn:rsuf1}
	\end{equation}
	
	Analogously, for (D-1-2) and for the first user, the achievable rate of symbols $s_{1,l}^{\mathrm{U}}, l=1,\dots,L_1,$ is given by
	\begin{equation}
	R_{1,l}^{{\mathrm{U}}, (2)} = \begin{cases}
	\log_2\left(1+\frac{\frac{1}{{\Pi_1}}p_{1,l}^{\mathrm{U}} (\rho^{(1)}_{l,l})^2}{\sigma^2 +  \frac{1}{{\Pi_2}}p_{2,\bar{M}_2+l}^{\mathrm{U}} (\rho^{(2)}_{\bar{M}_2+l,\bar{M}_2+l})^2}\right) & \text{for $l=1,\dots,M$} \\
	\log_2\left(1+ \frac{1}{{\Pi_1}}\frac{p_{2,l}^{\mathrm{U}} (\rho^{(1)}_{l,l})^2}{\sigma^2}\right) & \text{for $l=M+1,\dots,L_1.$}
	\end{cases} \label{eqn:rfuf1}
	\end{equation}
	
	For the second user, after SIC, the achievable rate of symbols $s_{2,l}^{\mathrm{U}}, l=1,\dots,L_2,$ is given by
	\begin{equation}
	\begin{array}{lr}
	R_{2,l}^{{\mathrm{U}}, (2)} = \log_2\left(1+\frac{1}{{\Pi_2}}\frac{p_{2,l}^{\mathrm{U}} (\rho^{(2)}_{l,l})^2}{\sigma^2}\right) & \text{for $l=1,\dots,L_2.$}
	\end{array} \label{eqn:rfuf2}
	\end{equation}
	
	\section{Proposed Downlink ST MIMO-NOMA Scheme}
	\label{sec:propdownlink}
	Analogous to Section \ref{sec:propuplink}, in this section, we begin by developing a matrix decomposition for ST of the downlink MIMO channels of the users, which we then exploit to design the proposed downlink ST MIMO-NOMA precoding and decoding schemes. Subsequently, the corresponding achievable rate expressions are determined.
	
	\subsection{Simultaneous Triangularization for Downlink Transmission}
	Let\footnote{In this section, we define $\bar{M}_1, \bar{M}_2,$ and $M$ for downlink transmission in terms of downlink symbol vector length $L.$ In contrast, in Section \ref{sec:propuplink}, they were defined for uplink transmission in terms of uplink symbol vector lengths $L_1$ and $L_2.$ Nevertheless, we emphasize that the values obtained for  $\bar{M}_1, \bar{M}_2,$ and $M$ in both cases are identical.} $\bar{M}_1 = \mathrm{max}\left\{0,\mathrm{min}\left\{M_1,L-M_2\right\}\right\}, \bar{M}_2 = \mathrm{max}\left\{0,\mathrm{min}\left\{M_2,L-M_1\right\}\right\},$ and $M = N-\bar{M}_1-\bar{M}_2.$ Downlink ST of $\boldsymbol{H}_1$ and $\boldsymbol{H}_2$ is compactly stated in the following theorem.
	
	\begin{theorem}
		\label{th:std}
		Let $\boldsymbol{H}_1$ and $\boldsymbol{H}_2$ be as defined in Section \ref{sec:cm}. Then, there exist unitary matrices $\boldsymbol{Q}_1^{\mathrm{D}} \in \mathbb{C}^{M_1\times M_1}, \boldsymbol{Q}_2^{\mathrm{D}} \in \mathbb{C}^{M_2\times M_2},$ and a full matrix $\boldsymbol{X}^{\mathrm{D}} \in \mathbb{C}^{N\times L}$ such that
		\begin{align}
		\boldsymbol{Q}_1^{\mathrm{D}}\boldsymbol{H}_1\boldsymbol{X}^{\mathrm{D}} &= \begin{bmatrix}\boldsymbol{R}_1^{\mathrm{D}} & \boldsymbol{0}\end{bmatrix}, \label{eqn:std1}\\
		\boldsymbol{Q}_2^{\mathrm{D}}\boldsymbol{H}_2\boldsymbol{X}^{\mathrm{D}} &= \begin{bmatrix}\boldsymbol{R}_2'^{\mathrm{D}} & \boldsymbol{0} & \boldsymbol{R}_2''^{\mathrm{D}}\end{bmatrix}, \label{eqn:std2}
		\end{align}
		where $\boldsymbol{R}_1^{\mathrm{D}} \in \mathbb{C}^{M_1\times (M+\bar{M}_1)}, \boldsymbol{R}_2'^{\mathrm{D}} \in \mathbb{C}^{M_2\times M},$ and $\boldsymbol{R}_2''^{\mathrm{D}} \in \mathbb{C}^{M_2\times \bar{M}_2}.$ Furthermore, $\boldsymbol{R}_1^{\mathrm{D}}$ and $\boldsymbol{R}_2^{\mathrm{D}} = \begin{bmatrix}\boldsymbol{R}_2'^{\mathrm{D}} & \boldsymbol{R}_2''^{\mathrm{D}}\end{bmatrix} \in \mathbb{C}^{M_2\times (M+\bar{M}_2)}$ are upper-triangular matrices with real-valued entries on their main diagonals.
	\end{theorem}
	\begin{proof}
		Please refer Appendix \ref{app:std}.
	\end{proof}	
	
	\subsection{Proposed Downlink Pre-processing Scheme}
	\label{sec:prec}
	Based on Theorem \ref{th:std}, the precoder matrix can be chosen as $\boldsymbol{P}^{\mathrm{D}} = \boldsymbol{X}^{\mathrm{D}},$ and the detection matrices of users 1 and 2 can be chosen directly as $\boldsymbol{Q}_1^{\mathrm{D}}$ and $\boldsymbol{Q}_2^{\mathrm{D}}$ for users 1 and 2, respectively\footnote{Note that although the proposed scheme utilizes QR decomposition based detection matrices, other detection schemes such as zero forcing or joint decoding can also be utilized for the proposed precoder $\boldsymbol{P}^{\mathrm{D}} = \boldsymbol{X}^{\mathrm{D}}.$}. Hence, the received signal at the users, based on (\ref{eqn:yk}), can be simplified to
	\begin{align}
	\tilde{\boldsymbol{y}}_k^{\mathrm{D}} &= \frac{1}{\sqrt{\Pi_k}} \boldsymbol{R}_k^{\mathrm{D}} \tilde{\boldsymbol{s}}_k^{\mathrm{D}} + \tilde{\boldsymbol{n}}_k^{\mathrm{D}} \label{eqn:ykr}
	\end{align}
	where, based on Theorem \ref{th:std}, symbol vectors $\tilde{\boldsymbol{s}}_k^{\mathrm{D}} \in \mathbb{C}^{(M+\bar{M}_k)\times 1}$ are defined as $[\tilde{\boldsymbol{s}}_k^{\mathrm{D}}]_l = s_l^{\mathrm{D}},$ $l=1,\dots,M,$ $[\tilde{\boldsymbol{s}}_1^{\mathrm{D}}]_{l+M} = s_{l+M}^{\mathrm{D}},$ $l=1,\dots,\bar{M}_1,$ and $[\tilde{\boldsymbol{s}}_2^{\mathrm{D}}]_{l+M} = s_{l+M+\bar{M}_1}^{\mathrm{D}},$ $l=1,\dots,\bar{M}_2,$ and $\tilde{\boldsymbol{n}}_k^{\mathrm{D}}$ contains the corresponding elements of $\boldsymbol{n}_k^{\mathrm{D}}.$
	
	As symbols $s_{l}^{\mathrm{D}}, l=M+1,\dots,M+\bar{M}_1,$ are only transmitted to user 1, $p_{2,l}^{\mathrm{D}} = 0$ for $l=M+1,\dots,M+\bar{M}_1.$ Similarly, $p_{1,l}^{\mathrm{D}} = 0$ for $l=M+\bar{M}_1+1,\dots,L.$ Power allocation coefficients $p_{k,l}^{\mathrm{D}}, k=1,2, l=1,\dots,M,$ can be used to adjust the rates of users 1 and 2. Lastly, the $M_k\times(M+\bar{M}_k)$ matrices $\boldsymbol{R}_k^{\mathrm{D}}, k=1,2,$ are of the form		
	\begin{equation}
	\begin{bmatrix}
		\varrho^{(k)}_{1,1} &  \dots   & \varrho^{(k)}_{1,M+\bar{M}_k}   \\
		0		 & \ddots &  \vdots                       \\
		\vdots        &   0      & \varrho^{(k)}_{M+\bar{M}_k,M+\bar{M}_k} \\
		\boldsymbol{0}    & \boldsymbol{0} &  \boldsymbol{0}\\
		\end{bmatrix}. \label{eqn:Rk}
	\end{equation}
	
	\subsection{Proposed Downlink Decoding Scheme}
	\label{sec:dec}
	Let $\hat{s}_{k,l}^{\mathrm{D}},$ $k=1,2,l=1,\dots,L,$ denote the detected symbols corresponding to transmitted symbols $s_{k,l}^{\mathrm{D}}.$ As the rates of $s_{k,l}^{\mathrm{D}},$ $k=1,2,l=1,\dots,L,$ are chosen such that they lie within the achievable rate region, perfect decoding, i.e., $\hat{s}_{k,l}^{\mathrm{D}} = s_{k,l}^{\mathrm{D}},$ is assumed in the following.
	
	From the upper triangular structure of $\boldsymbol{R}_k^{\mathrm{D}}$ in (\ref{eqn:Rk}), and based on (\ref{eqn:std1}) and (\ref{eqn:std2}), we note that the symbols corresponding to the last $\bar{M}_k$ columns of $\boldsymbol{R}_k^{\mathrm{D}}, k=1,2,$ which are only transmitted to user $k,$ contain no inter-user-interference and are therefore decoded directly, in reverse order. Next, the symbols corresponding to the first $M$ columns of $\boldsymbol{R}_k^{\mathrm{D}}, k=1,2,$ which are transmitted to both users, are decoded as in SISO-NOMA \cite{Saito2013}, also in reverse order, see Figure \ref{fig:std}. For each symbol, the self-interference of the previously decoded symbols is eliminated. The decoding process is described in detail below.

	\begin{figure*}
		\centering
		\begin{minipage}{0.5\textwidth}
			\centering
			\includegraphics[width=\textwidth]{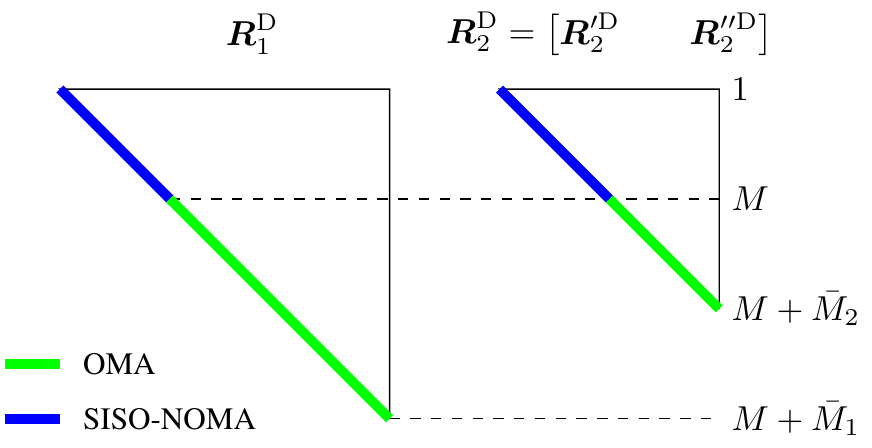}
			\caption{Simultaneous triangularization in downlink ST MIMO-NOMA for $\bar{M}_1,\bar{M}_2,M > 0.$ The green and blue diagonal portions are decoded successively using conventional OMA and SISO-NOMA decoding, respectively.}
			\label{fig:std}
		\end{minipage}\hspace{10pt}%
		\begin{minipage}{0.45\textwidth}
			\centering
			\captionof{table}{Complexity of precoder and detection matrix computation for the proposed downlink ST MIMO-NOMA and the considered MIMO-NOMA baseline schemes.}
			\label{tab:d}
			\begin{tabular}{|p{0.55\textwidth}|p{0.3\textwidth}|}
				\hline
				MIMO-NOMA Scheme & Complexity	 \\\hline\hline
				Proposed ST & $\mathcal{O}\mkern-\medmuskip
\left(6 N^3\right)$ \\\hline
				SD in \cite{Krishnamoorthy2020} & $\mathcal{O}\mkern-\medmuskip
\left(44.66 N^3\right)$ \\\hline
				SD in \cite{Chen2019} & $\mathcal{O}\mkern-\medmuskip
\left(38.66 N^3\right)$ \\\hline
			\end{tabular}
		\end{minipage}
	\end{figure*}
	
	The first user decodes the symbols as follows.
	\begin{enumerate}
		\item If $\bar{M}_1 > 0,$ symbols $s_{1,l}^{\mathrm{D}},$ $l=M+1,\dots,M+\bar{M}_1,$ are decoded, in reverse order, starting from the self-interference free element $[\tilde{\boldsymbol{y}}_1^{\mathrm{D}}]_{M+\bar{M}_1}$ given in (\ref{eqn:ykr}). For each subsequent symbol, the self-interference from the previously decoded symbols\footnote{Similar to the uplink case, the assumption of perfect decoding is justified as the rates corresponding to the decoded symbols are at or below the achievable rate $R_{k,l}$ provided in Section \ref{sec:drates}.} is eliminated, resulting in the self-interference free signal
		\begin{align}
		[\hat{\boldsymbol{y}}_1^{\mathrm{D}}]_{l} = [\tilde{\boldsymbol{y}}_1^{\mathrm{D}}]_{l} - \frac{1}{\sqrt{\Pi_1}} \sum_{l'=l+1}^{M+\bar{M}_1} \varrho^{(1)}_{l,l'}\sqrt{p_{1,l'}^{\mathrm{D}}}\hat{s}_{1,l'}^{\mathrm{D}} = \frac{1}{\sqrt{\Pi_1}} \varrho^{(1)}_{l,l} \sqrt{p_{1,l}^{\mathrm{D}}} s_{1,l}^{\mathrm{D}} + [\tilde{\boldsymbol{n}}_k^{\mathrm{D}}]_l, \label{eqn:dec12}
		\end{align}
		for $l=M+1,\dots,M+\bar{M}_1$ (in reverse order), which is decoded.
		
		\item Next, if $M > 0,$ symbols $s_{1,l}^{\mathrm{D}},$ $l=1,\dots,M,$ which contain inter-user-interference from $s_{2,l}^{\mathrm{D}},l=1,\dots,M,$ are decoded directly as in SISO-NOMA, in reverse order. Then, as above, for each symbol, the self-interference from the previously decoded symbols is eliminated, resulting in the self-interference free signal
		\begin{align}
		[\hat{\boldsymbol{y}}_1^{\mathrm{D}}]_{l} &= [\tilde{\boldsymbol{y}}_1^{\mathrm{D}}]_{l} - \frac{1}{\sqrt{\Pi_1}} \sum_{l'=l+1}^{M+\bar{M}_1} \varrho^{(1)}_{l,l'}\sqrt{p_{1,l'}^{\mathrm{D}}} \hat{s}_{1,l'}^{\mathrm{D}} \nonumber\\
		&= \frac{1}{\sqrt{\Pi_1}} \varrho^{(1)}_{l,l} \sqrt{p_{1,l}^{\mathrm{D}}} s_{1,l}^{\mathrm{D}} + \frac{1}{\sqrt{\Pi_1}} \sum_{l'=l}^{M} \varrho^{(1)}_{l,l'}\sqrt{p_{2,l'}^{\mathrm{D}}} s_{2,l'}^{\mathrm{D}} + [\tilde{\boldsymbol{n}}_k^{\mathrm{D}}]_l, \label{eqn:dec11}
		\end{align}
		for $l=1,\dots,M$ (in reverse order), which is decoded. Note that the residual inter-user-interference from symbols $s_{2,l}^{\mathrm{D}},l=1,\dots,M,$ cannot be eliminated and is treated as noise.
		
	\end{enumerate}
	Analogously, the second user decodes the symbols as follows.
	\begin{enumerate}
		\item If $\bar{M}_2 > 0,$ symbols $s_{2,l}^{\mathrm{D}},$ $l=M+\bar{M}_1+1,\dots,M+\bar{M}_1+\bar{M}_2,$ are decoded, in reverse order, starting from the self-interference free element $[\tilde{\boldsymbol{y}}_2^{\mathrm{D}}]_{M+\bar{M}_2}$ given in (\ref{eqn:ykr}). For each symbol, the self-interference from the previously decoded symbols is eliminated, resulting in the self-interference free signal
		\begin{align}
		[\hat{\boldsymbol{y}}_2^{\mathrm{D}}]_{l-\bar{M}_1} &= [\tilde{\boldsymbol{y}}_2^{\mathrm{D}}]_{l-\bar{M}_1} -  \frac{1}{\sqrt{\Pi_2}} \sum_{l'=l+1}^{M+\bar{M}_1+\bar{M}_2} \varrho^{(2)}_{l-\bar{M}_1,l'-\bar{M}_1}\sqrt{p_{2,l'}^{\mathrm{D}}}\hat{s}_{2,l'}^{\mathrm{D}} \nonumber\\
		&= \frac{1}{\sqrt{\Pi_2}} \varrho^{(2)}_{l-\bar{M}_1,l-\bar{M}_1} \sqrt{p_{2,l}^{\mathrm{D}}} s_{2,l}^{\mathrm{D}} + [\tilde{\boldsymbol{n}}_k^{\mathrm{D}}]_{l-\bar{M}_1}, \label{eqn:dec22}
		\end{align}
		for $l=M+\bar{M}_1+1,\dots,M+\bar{M}_1+\bar{M}_2$ (in reverse order), which is decoded.
		
		\item Next, if $M > 0,$ symbols $s_{k,l}^{\mathrm{D}},$ $k=1,2,l=1,\dots,M,$ are decoded as in SISO-NOMA \cite{Saito2013} where the first user's symbols are decoded directly and the second user's symbols are decoded after SIC, in reverse order. In this case also, for each symbol, the self-interference from the previously decoded symbols is eliminated, resulting in the self-interference free signal
		\begin{align}
		[\hat{\boldsymbol{y}}_2^{\mathrm{D}}]_{l} &= [\tilde{\boldsymbol{y}}_2^{\mathrm{D}}]_{l} - \frac{1}{\sqrt{\Pi_2}} \sum_{l'=l+1}^{M} \varrho^{(2)}_{l,l'}(\sqrt{p_{1,l'}^{\mathrm{D}}}\hat{s}_{1,l'}^{\mathrm{D}}+\sqrt{p_{2,l'}^{\mathrm{D}}}\hat{s}_{2,l'}^{\mathrm{D}}) - \frac{1}{\sqrt{\Pi_2}} \sum_{l'=M+\bar{M}_1+1}^{M+\bar{M}_1+\bar{M}_2} \varrho^{(2)}_{l,l'-\bar{M}_1}\sqrt{p_{2,l'}^{\mathrm{D}}}\hat{s}_{2,l'}^{\mathrm{D}} \nonumber\\
		&= \frac{1}{\sqrt{\Pi_2}} \varrho^{(2)}_{l,l} \left(\sqrt{p_{1,l}^{\mathrm{D}}}s_{1,l}^{\mathrm{D}} + \sqrt{p_{2,l}^{\mathrm{D}}}s_{2,l}^{\mathrm{D}}\right) + [\tilde{\boldsymbol{n}}_k^{\mathrm{D}}]_l, \label{eqn:dec21}
		\end{align}
		for $l=1,\dots,M$ (in reverse order), which is decoded as in SISO-NOMA.
	\end{enumerate}
	
	As seen from the decoding scheme above, the MIMO-NOMA channels of the users are decomposed into scalar channels. Furthermore, the second user is always the SIC user, i.e., no SIC capability is necessary at the first 
	user. Therefore, unlike the uplink case with two decoding orders, in this case, only decoding order (D-1-2) at user 2 is utilized.
		
	\begin{remark}
		From (\ref{eqn:dec12}) and (\ref{eqn:dec22}), we observe that the proposed MIMO-NOMA precoding scheme exploits the available null spaces of the MIMO channel matrices of the users for transmitting $\bar{M}_1$ and $\bar{M}_2$ symbols inter-user-interference free to users 1 and 2, respectively. Furthermore, from (\ref{eqn:x}), we note that the precoder matrix avoids inversion of the MIMO channels of the users.
	\end{remark}	

	\subsection{Computational Complexity for Downlink}
	\label{sec:dcomp}
	For an $m\times n$ matrix $\boldsymbol{X}$, a matrix $\bar{\boldsymbol{X}}$ containing a basis for the null space can be computed using the QR decomposition, which entails a complexity of $\mathcal{O}\mkern-\medmuskip
\left(2 n m^2\right)$ \cite[Alg. 5.2.5]{Golub2012}. Hence, constructing the precoder matrix in (\ref{eqn:x}) entails a total complexity of $\mathcal{O}\mkern-\medmuskip
\left(2 N M_1^2 + 2 N M_2^2 + 2 N (N-M)^2\right)$ at the BS for computing $\bar{\boldsymbol{H}}_1, \bar{\boldsymbol{H}}_2,$ and $\boldsymbol{K}$ via the QR decomposition. Next, at user $k,$ $k=1,2,$ detection matrix $\boldsymbol{Q}_k,$ which is also computed using the QR decomposition, and self-interference cancellation entail complexities of $\mathcal{O}\mkern-\medmuskip
\left(2 M_k^3\right)$ and $\mathcal{O}\mkern-\medmuskip
\left(M_k^2\right),$ respectively, resulting in a total complexity of $\mathcal{O}\mkern-\medmuskip
\left(6 N^3\right)$ for $M_1=M_2=N.$ For comparison, the SD schemes in \cite{Krishnamoorthy2020} and \cite{Chen2019} entail complexities of $\mathcal{O}\mkern-\medmuskip
\left(44.66 N^3\right)$ and $\mathcal{O}\mkern-\medmuskip
\left(38.66 N^3\right),$ respectively \cite[Sec. III-D]{Krishnamoorthy2020}, as summarized in Table \ref{tab:d}. Therefore, the proposed scheme and the SD schemes in \cite{Chen2019} and \cite{Krishnamoorthy2020} entail an identical overall worst case complexity of $\mathcal{O}\mkern-\medmuskip
\left(N^3\right).$
	
	\subsection{Downlink Achievable Rates}
	\label{sec:drates}
	At the first user, based on (\ref{eqn:dec11}), the achievable rate for $s_{1,l}^{\mathrm{D}}, l=1,\dots,L,$ after self-interference cancellation, is given by
	\begin{align}
	R_{1,l}^{{\mathrm{D}},(1)} = \log_2\left(1 + \frac{\frac{1}{\Pi_1}p_{1,l}^{\mathrm{D}}\left|\varrho^{(1)}_{l,l}\right|^2}{\sigma^2 + \frac{1}{\Pi_1}\sum_{l'=l}^{M}p_{2,l'}^{\mathrm{D}}\left|\varrho^{(1)}_{l,l'}\right|^2}\right), \label{eqn:r1lm}
	\end{align}
	for $l = 1,\dots,M,$ and, based on (\ref{eqn:dec12}), by
	\begin{align}
	R_{1,l}^{{\mathrm{D}}, (1)} = \log_2\left(1 + \frac{1}{\Pi_1}\frac{p_{1,l}^{\mathrm{D}}\left|\varrho^{(1)}_{l,l}\right|^2}{\sigma^2}\right),
	\end{align}
	for $l = M+1,\dots,M+\bar{M}_1$. Furthermore, as mentioned in Section \ref{sec:prec}, $R_{1,l}^{{\mathrm{D}}, (1)}=0$ for $l=M+\bar{M}_1+1,\dots,L.$
	
	Similarly, at the second user, the achievable rate for $s_{1,l}^{\mathrm{D}}, l=1,\dots,M,$ after self-interference cancellation, based on (\ref{eqn:dec21}), is given by
	\begin{align}
	R_{1,l}^{{\mathrm{D}}, (2)} = \log_2\left(1 + \frac{\frac{1}{\Pi_2}p_{1,l}^{\mathrm{D}}\left|\varrho^{(2)}_{l,l}\right|^2}{\sigma^2 + \frac{1}{\Pi_2}p_{2,l}^{\mathrm{D}}\left|\varrho^{(2)}_{l,l}\right|^2}\right). \label{eqn:r1lm2}
	\end{align}
	Note that the difference to (\ref{eqn:r1lm}) arises because the inter-user-interference from symbols $s_{2,l}^{\mathrm{D}}, l=1,\dots,M,$ can be eliminated at the second user as both users' symbols are decoded as in SISO-NOMA.
	
	In order to ensure that symbols $s_{1,l}^{\mathrm{D}}, l=1,\dots,M,$ can be decoded at both users, the rates are chosen as the instantaneous minimum rate. Hence,
	\begin{align}
	R^{\mathrm{D}}_{1,l} = \mathrm{min}\left\{R_{1,l}^{{\mathrm{D}}, (1)}, R_{1,l}^{{\mathrm{D}}, (2)}\right\}, \label{eqn:r1l}
	\end{align}
	for $l=1,\dots,M,$ and $R^{\mathrm{D}}_{1,l} = R_{1,l}^{{\mathrm{D}}, (1)}$ for $l=M+1,\dots,L.$
	
	Lastly, at the second user, based on (\ref{eqn:dec21}), for $l=1,\dots,M,$ the achievable rate after SIC and self-interference cancellation is given by
	\begin{align}
	R^{\mathrm{D}}_{2,l} = \log_2\left(1 + \frac{1}{\Pi_2}\frac{p_{2,l}^{\mathrm{D}}\left|\varrho^{(2)}_{l,l}\right|^2}{\sigma^2}\right), \label{eqn:r2lm}
	\end{align}
	and, based on (\ref{eqn:dec22}), for $l=M+\bar{M}_1+1,\dots,L,$ the achievable rate after self-interference cancellation is given by
	\begin{align}
	R^{\mathrm{D}}_{2,l} = \log_2\left(1 + \frac{1}{\Pi_2}\frac{p_{2,l}^{\mathrm{D}}\left|\varrho^{(2)}_{l-\bar{M}_1,l-\bar{M}_1}\right|^2}{\sigma^2}\right).
	\end{align}
	Furthermore, as mentioned in Section \ref{sec:prec}, $R^{\mathrm{D}}_{2,l}=0$ for $l=M+1,\dots,M+\bar{M}_1.$
	\section{Maximum Achievable Rate Regions}
	\label{sec:opa}
	In this section, we characterize the maximum achievable rate regions of the proposed uplink and downlink ST MIMO-NOMA schemes.
	
	\subsection{Uplink}
	\label{sec:uopt}
	Based on the achievable rates given in Section \ref{sec:urates}, the optimal power allocation for the proposed uplink scheme for (D-2-1) and (D-1-2) can be obtained by solving the following optimization problem:
	\begin{maxi!}
		{p_{k,l}^{\mathrm{U}} \geq 0\,\forall\,k,l}
		{\sum_{l=1}^{L} \left[R_{1,l}^{{\mathrm{U}},(d)} + R_{2,l}^{{\mathrm{U}},(d)}\right]}
		{\label{opt:u1}}
		{}
		\addConstraint{\sum_{l=1}^{L} p_{1,l}^{\mathrm{U}}}{\leq P_1^{\mathrm{U}}}{}
		\addConstraint{\sum_{l=1}^{L} p_{2,l}^{\mathrm{U}}}{\leq P_2^{\mathrm{U}}}{},
	\end{maxi!}
	for $d,k=1,2,$ respectively. As $R_{1,l}^{{\mathrm{U}},(d)} + R_{2,l}^{{\mathrm{U}},(d)}$ is concave for all $d,l,$ problem (\ref{opt:u1}) is a convex optimization problem which can be solved efficiently with standard optimization techniques \cite{Boyd2004}. Furthermore, as (\ref{opt:u1}) is convex, the solution is globally optimal \cite{Rockafellar1970}.

	The achievable rate region of the proposed uplink scheme is given by the 2-dimensional convex polytope with corner points $(0,0), (0,R_2^{{\mathrm{U}},2\star}), (R_1^{{\mathrm{U}},2\star}, R_2^{{\mathrm{U}},2\star}), (R_1^{{\mathrm{U}},1\star}, R_2^{{\mathrm{U}},1\star}),$ and $(R_1^{{\mathrm{U}},1\star}, 0)$ \cite[Chap. 4]{Gamal2011}, where $R_1^{{\mathrm{U}},d\star} = \sum_{l=1}^{L}R_{1,l}^{{\mathrm{U}},d\star}$ and $R_2^{{\mathrm{U}},d\star} = \sum_{l=1}^{L}R_{2,l}^{{\mathrm{U}},d\star}, d=1,2,$ denote the optimal rates of users 1 and 2, respectively, which are obtained by solving (\ref{opt:u1}).

	\subsection{Downlink}
	\label{sec:dopt}
	Based in the achievable rates given in Section \ref{sec:drates}, the optimal power allocation for the proposed downlink scheme can be found by solving the following optimization problem:
	\begin{maxi!}
		{p_{k,l}^{\mathrm{D}} \geq 0\,\forall\,k,l}
		{\sum_{l=1}^{L} \left[\eta R_{1,l}^{{\mathrm{D}}} + (1-\eta) R_{2,l}^{{\mathrm{D}}}\right]}
		{\label{opt:d1}}
		{R_{\eta}^{{\mathrm{D}}\star}=}
		\addConstraint{\sum_{l=1}^{L} (p_{1,l}^{\mathrm{D}}+ p_{2,l}^{\mathrm{D}})}{\leq P_\mathrm{T}^{\mathrm{D}}}{},
	\end{maxi!}
	where $\eta \in [0,1]$ is a fixed weight which can be chosen to adjust the achievable rates of user 1 and 2 during power allocation \cite[Sec. 4]{WangGiannakis2011}.
	
	However, unlike the analogous optimization problem for the uplink, the problem in (\ref{opt:d1}) is non-convex due to the coupling of the transmit powers $p_{k,l}^{\mathrm{D}},k=1,2,l=1,\dots,M,$ in $R_{1,l}^{{\mathrm{D}}}.$ Hence, in the following, we simplify (\ref{opt:d1}) to obtain a tractable characterization of the achievable rate region.
	
	To this end, first, we rewrite (\ref{eqn:r1l}) as follows:
	\begin{align}
		R^{\mathrm{D}}_{1,l} = \begin{cases}
		R_{1,l}^{{\mathrm{D}}, (2)} & \text{$\frac{1}{\Pi_2}\left|\varrho^{(2)}_{l,l}\right|^2 +  \frac{1}{\Pi_1\Pi_2\sigma^2}\left|\varrho^{(2)}_{l,l}\right|^2\sum_{l'=l+1}^{M}p_{2,l'}^{\mathrm{D}}\left|\varrho^{(1)}_{l,l'}\right|^2 \leq \frac{1}{\Pi_1}\left|\varrho^{(1)}_{l,l}\right|^2$} \\
		R_{1,l}^{{\mathrm{D}}, (1)} & \text{otherwise,}
		\end{cases} \label{eqn:minsimp}
	\end{align}
	for $l=1,\dots,M,$ which follows directly from the definition of $R_{1,l}^{{\mathrm{D}}, (1)}$ and $R_{1,l}^{{\mathrm{D}}, (2)}$ in (\ref{eqn:r1lm}) and (\ref{eqn:r1lm2}), respectively. The condition in (\ref{eqn:minsimp}) depends on the power allocation coefficients $p_{2,l}^{\mathrm{D}},l=1,\dots,M.$ Additionally, the sum $\eta R_{1,l}^{{\mathrm{D}}, (2)} + (1-\eta) R_{2,l}^{{\mathrm{D}}}$ is concave while $\eta R_{1,l}^{{\mathrm{D}}, (1)} + (1-\eta)  R_{2,l}^{{\mathrm{D}}}$ is not concave. These two properties render a concave transformation of the objective function difficult. Hence, in the following, we characterize the global optimum of (\ref{opt:d1}) via an upper and a lower bound. As shown later in Section \ref{sec:dsim}, the gap between the upper and lower bounds is negligibly small.

	\subsubsection{An Upper and a Lower Bound}
	Based on (\ref{eqn:minsimp}), an upper ($\uparrow$) and a lower ($\downarrow$) bound for $R^{\mathrm{D}}_{1,l}, l=1,\dots,L,$ can be obtained as follows:
	\begin{align}
		R^{\mathrm{D\uparrow}}_{1,l} &= R_{1,l}^{{\mathrm{D}}, (1)}, \\
		R^{\mathrm{D\downarrow}}_{1,l} &= \begin{cases}
			0 & \text{if $l < M$ and $\frac{1}{\Pi_2}\left|\varrho^{(2)}_{l,l}\right|^2 \leq \frac{1}{\Pi_1}\left|\varrho^{(1)}_{l,l}\right|^2$} \\
			R_{1,l}^{{\mathrm{D}}, (2)} & \text{if $l = M$ and $\frac{1}{\Pi_2}\left|\varrho^{(2)}_{l,l}\right|^2 \leq \frac{1}{\Pi_1}\left|\varrho^{(1)}_{l,l}\right|^2$} \\
			R_{1,l}^{{\mathrm{D}}, (1)} & \text{otherwise.}
		\end{cases} \label{eqn:bounds}
	\end{align}
	Hence, new optimization problems (\ref{opt:d1}$\uparrow$) and (\ref{opt:d1}$\downarrow$), with optimal values denoted by $R_{\eta}^{\mathrm{D\uparrow}\star}$ and $R_{\eta}^{\mathrm{D\downarrow}\star},$ can be obtained by replacing $R_{1,l}^{{\mathrm{D}}}$ in (\ref{opt:d1}) with $R^{\mathrm{D\uparrow}}_{1,l}$ and $R^{\mathrm{D\downarrow}}_{1,l},l=1,\dots,L,$ respectively. We have, $R_{\eta}^{\mathrm{D\downarrow}\star} \leq R_{\eta}^{{\mathrm{D}}\star} \leq R_{\eta}^{\mathrm{D\uparrow}\star}\,\forall\,\eta.$
	
	In the following, first, we focus on solving (\ref{opt:d1}$\uparrow$). The solution of (\ref{opt:d1}$\downarrow$) is given in Section \ref{sec:lb}.
		
	\subsubsection{Problem Transformation}
	In order to simplify non-convex problem (\ref{opt:d1}$\uparrow$), we utilize the following theorem.
		
	\begin{theorem}
		\label{th:trans}
		The non-concave objective function
		\begin{align}
			R_{\eta}^{\mathrm{D\uparrow}} = \sum_{l=1}^{L} \left[\eta  R_{1,l}^{{\mathrm{D}}\uparrow} + (1-\eta) R_{2,l}^{{\mathrm{D}}}\right] \label{eqn:bc}
		\end{align}
		in $p_{k,l}^{\mathrm{D}},k=1,2,l=1,\dots,L,$ can be equivalently reformulated into concave function
		\begin{align}
			\scalebox{0.9}{\mbox{\ensuremath{\displaystyle R_{\eta}^{\prime\mathrm{D\uparrow}}}}} &= \scalebox{0.9}{\mbox{\ensuremath{\displaystyle \sum_{l=1}^{M}\left[(1-2\eta) \log_2\left(1 + \frac{1}{\Pi_1\sigma^2} \left|\varrho^{(1)}_{l,l}\right|^2 p_{1,l}'^{\mathrm{D}}\right) + \eta \log_2\left(1 + \frac{1}{\Pi_1\sigma^2} \left|\varrho^{(1)}_{l,l}\right|^2 p_{1,l}'^{\mathrm{D}} + \frac{1}{\Pi_2\sigma^2} \left|\varrho^{(2)}_{l,l}\right|^2 p_{2,l}'^{\mathrm{D}}\right)\right]}}} \nonumber\\
			&\quad \scalebox{0.85}{\mbox{\ensuremath{\displaystyle {}+  \eta\sum_{l=M+1}^{M+\bar{M}_1}\log_2\left(1 + \frac{1}{\Pi_1\sigma^2}p_{1,l}'^{\mathrm{D}}\left|\varrho^{(1)}_{l,l}\right|^2\right) + (1-\eta) \sum_{l=M+\bar{M}_1+1}^{M+\bar{M}_1+\bar{M}_2} \log_2\left(1 + \frac{1}{\Pi_2\sigma^2}p_{2,l}'^{\mathrm{D}}\left|\varrho^{(2)}_{l-\bar{M}_1,l-\bar{M}_1}\right|^2\right)}}},\label{eqn:mac1}
		\end{align}	
		if $\eta \leq 0.5,$ and into concave function
		\begin{align}
			\scalebox{0.87}{\mbox{\ensuremath{\displaystyle R_{\eta}^{\prime\mathrm{D\uparrow}}}}} &= \scalebox{0.87}{\mbox{\ensuremath{\displaystyle \sum_{l=1}^{M}\left[(2\eta-1) \log_2\left(1 + \frac{1}{\Pi_2\sigma^2} \left|\varrho^{(2)}_{l,l}\right|^2 p_{2,l}'^{\mathrm{D}}\right) + (1-\eta) \log_2\left(1 + \frac{1}{\Pi_1\sigma^2} \left|\varrho^{(1)}_{l,l}\right|^2 p_{1,l}'^{\mathrm{D}} + \frac{1}{\Pi_2\sigma^2} \left|\varrho^{(2)}_{l,l}\right|^2 p_{2,l}'^{\mathrm{D}}\right)\right]}}} \nonumber\\
			&\quad \scalebox{0.85}{\mbox{\ensuremath{\displaystyle {}+  \eta\sum_{l=M+1}^{M+\bar{M}_1}\log_2\left(1 + \frac{1}{\Pi_1\sigma^2}p_{1,l}'^{\mathrm{D}}\left|\varrho^{(1)}_{l,l}\right|^2\right) + (1-\eta) \sum_{l=M+\bar{M}_1+1}^{M+\bar{M}_1+\bar{M}_2} \log_2\left(1 + \frac{1}{\Pi_2\sigma^2}p_{2,l}'^{\mathrm{D}}\left|\varrho^{(2)}_{l-\bar{M}_1,l-\bar{M}_1}\right|^2\right)}}}, \label{eqn:mac2}
		\end{align}
		if $\eta > 0.5,$ such that $R_{\eta}^{\mathrm{D\uparrow}}=R_{\eta}^{\prime\mathrm{D\uparrow}}$ and
		\begin{align}
			\sum_{l=1}^{L} (p_{1,l}^{\mathrm{D}} + p_{2,l}^{\mathrm{D}}) = \sum_{l=1}^{L} (p_{1,l}'^{\mathrm{D}} + p_{2,l}'^{\mathrm{D}}) + \sum_{l'=2}^{M}\frac{p_{2,l'}'^{\mathrm{D}}}{\sigma^2 + \frac{1}{\Pi_1}p_{1,l'}'^{\mathrm{D}}\left|\varrho^{(1)}_{l',l'}\right|^2 }\sum_{l=1}^{l'-1}\frac{1}{\Pi_1} p_{1,l}'^{\mathrm{D}} \left|\varrho^{(1)}_{l,l'}\right|^2, \label{eqn:pe}
		\end{align}
		where
		\begin{align}
		p_{1,l}'^{\mathrm{D}} &= \begin{cases}
		\frac{p_{1,l}^{\mathrm{D}}\sigma^2}{\sigma^2 + \frac{1}{\Pi_1}\sum_{l'=l}^{M}p_{2,l'}^{\mathrm{D}}\left|\varrho^{(1)}_{l,l'}\right|^2} & \text{if $l \leq M$} \\
		p_{1,l}^{\mathrm{D}} & \text{otherwise,}
		\end{cases} \label{eqn:p1t}\\
		p_{2,l}'^{\mathrm{D}} &= \begin{cases}
		\frac{p_{2,l}^{\mathrm{D}}}{\sigma^2}\left(\sigma^2 + \frac{1}{\Pi_1} p_{1,l}'^{\mathrm{D}}\left|\varrho^{(1)}_{l,l}\right|^2\right) & \text{if $l \leq M$} \\
		p_{2,l}^{\mathrm{D}} & \text{otherwise.}
		\end{cases} \label{eqn:p2t}
		\end{align}
	\end{theorem}
	\begin{proof}
		Please refer Appendix \ref{app:trans}.
	\end{proof}

	\begin{remark}
		Note that $p_{k,l}^{\mathrm{D}},k=1,2,l=1,\dots,L,$ can be recovered recursively from $p_{k,l}'^{\mathrm{D}},k=1,2,l=1,\dots,L,$ based on (\ref{eqn:p1t}) and  (\ref{eqn:p2t}).
	\end{remark}

	Now, based on Theorem \ref{th:trans}, (\ref{opt:d1}$\uparrow$) can be rewritten as follows.
	\begin{maxi!}
		{p_{k,l}'^{\mathrm{D}} \geq 0\,\forall\,k,l}
		{R_{\eta}^{\prime\mathrm{D\uparrow}}}
		{\label{opt:d1r}}
		{R_{\eta}^{\mathrm{D\uparrow}\star}=}
		\addConstraint{g_0 + \underbrace{\sum_{l'=2}^{M}g_{1,l'}g_{2,l'}g_{3,l'}}_{\coloneqq P_L}\label{optc:d1r}}{\leq P_\mathrm{T}}{},
	\end{maxi!}
	where
	\begin{align}
		g_0 &= \sum_{l=1}^{L} (p_{1,l}'^{\mathrm{D}}+ p_{2,l}'^{\mathrm{D}}), \\
		g_{1,l'} &= p_{2,l'}'^{\mathrm{D}},\quad	g_{2,l'} = \sum_{l=1}^{l'-1}\frac{1}{\Pi_1}  p_{1,l}'^{\mathrm{D}} \left|\varrho^{(1)}_{l,l'}\right|^2, \quad g_{3,l'} = \frac{1}{\sigma^2 + \frac{1}{\Pi_1}p_{1,l'}'^{\mathrm{D}}\left|\varrho^{(1)}_{l',l'}\right|^2}.
	\end{align}
	
	\begin{remark}
		Based on Theorem \ref{th:trans}, the BC in (\ref{opt:d1}$\uparrow$), in which the $l$-th symbol of user 1, $l=1,\dots,M,$ experiences inter-user-interference from all $l', l \leq l' \leq M,$ symbols of user 2, is transformed into a MAC in which the $l$-th symbol of user 1 experiences inter-user-interference only from the $l$-th symbol of user 2. However, communication over the resulting MAC incurs a power penalty of $P_L \geq 0.$
	\end{remark}

	\begin{remark}
		As $P_L$ is inversely proportional to $\Pi_1,$ the power penalty is insignificant when user 1 is located sufficiently far from the BS. This is the most relevant case for MIMO-NOMA as the loss incurred in the achievable rate if user 1 is not able to decode and cancel the interference from the symbols of user 2 becomes negligible.
	\end{remark}

	\begin{remark}
		The BC-MAC transformation in (\ref{eqn:p1t}) and (\ref{eqn:p2t}) is analogous to that introduced in \cite{Jindal2004} but is adapted to the problem at hand.
	\end{remark}

	\subsubsection{A Globally Optimal Solution}
	\label{sec:gopt}
	Problem (\ref{opt:d1r}) has a concave objective function and a monotonically increasing constraint function (\ref{optc:d1r}) in each auxiliary variable $g_0,g_{s,l'},s=1,2,3,l'=2,\dots,M,$ which are convex functions in $p_{k,l}'^{\mathrm{D}},k=1,2,l=1,\dots,M.$ Hence, the globally optimal solution of (\ref{opt:d1r}) lies on the boundary of the feasible set defined by (\ref{optc:d1r}). Therefore, (\ref{opt:d1r}) is a monotonic optimization problem and, in the following, we present an efficient algorithm for solving it based on \cite{Tuy2000}. The proposed approach utilizes a branch-and-bound iterative technique based on polyblocks \cite[Sec. 4]{Tuy2000a} for solving problems with a single non-convex but monotonically increasing constraint.
	
	We provide an outline of the algorithm as follows. According to \cite{Tuy2000}, the boundary of the feasible set defined by (\ref{optc:d1r}) is first loosely upper bounded via a polyblock utilizing feasible initial values for the individual auxiliary variables $g_0, g_{s,l'},s=1,2,3,l'=2,\dots,M.$ Subsequently, the upper bound is progressively tightened by shrinking the initial polyblock iteratively. In each step of the iteration, the polyblock corner point which maximizes (\ref{opt:d1r}) is chosen for further tightening. When the chosen corner point already lies on the boundary of the feasible set defined by (\ref{optc:d1r}), whereby no further tightening is possible, the algorithm is deemed to have converged. In the following, the algorithm is presented in detail.

	First, we initialize a $3(M-1)+1$ dimensional tuple of auxiliary variables $\boldsymbol{y} = (y_0, y_{s,l'},s=1,2,3,l'=2,\dots,M)$ to obtain an upper bound for $g_0, g_{s,l'},s=1,2,3,l'=2,\dots,M,$ respectively, as follows\footnote{The presented upper bound is exemplary. Convergence of the algorithm can be sped up by utilizing tighter upper bounds obtained with more sophisticated algorithms of higher complexity, see \cite[Prop. 2]{Tuy2000}.}:
	\begin{align}
		y_0^{(0)} = P_\mathrm{T}^{\mathrm{D}}, \quad y_{1,l'}^{(0)} = P_\mathrm{T}^{\mathrm{D}}, \quad y_{2,l'}^{(0)} = \mathrm{max}\left\{\frac{1}{\Pi_1} \left|\varrho^{(1)}_{l,l'}\right|^2, l=1,\dots,M\right\} P_\mathrm{T}^{\mathrm{D}}, \quad
		y_{3,l'}^{(0)} = \frac{1}{\sigma^2}, \label{eqn:y0}
	\end{align}
	which yields $\boldsymbol{y}^{(0)}.$ Next, we initialize the set $\mathcal{T},$ which contains the polyblock corner points, to $\mathcal{T}^{(0)} = \{\boldsymbol{y}^{(0)}\},$ and via line search, find a point 
	\begin{align}
	\boldsymbol{z}^{(0)} = \mu^{(0)} \boldsymbol{y}^{(0)}, \label{eqn:linesearch}
	\end{align}
	such that $\mu^{(0)} = \arg\max_{\alpha}\{\alpha \mid \alpha \geq 0 \text{ and } \psi(\alpha \boldsymbol{y}^{(0)}) \leq P_\mathrm{T}^{\mathrm{D}}\},$ where
	\begin{align}
		\psi(\boldsymbol{y}) = y_0 + \sum_{l'=2}^{M}y_{1,l'}y_{2,l'}y_{3,l'}.
	\end{align}
	
	In each iteration, $n=1,2,\dots,$ the set $\mathcal{T}^{(n-1)} \setminus \boldsymbol{y}^{(n-1)}$ is extended with points $\boldsymbol{y}^{(n-1),i}, i=1,\dots,3(M-1)+1,$ which are obtained by replacing the $i$-th element of $\boldsymbol{y}^{(n-1)}$ with that of $\boldsymbol{z}^{(n-1)},$ to obtain $\mathcal{T}^{(n)},$ i.e., 
	\begin{align}
		\mathcal{T}^{(n)} = \mathcal{T}^{(n-1)} \setminus \boldsymbol{y}^{(n-1)} \cup \{\boldsymbol{y}^{(n-1),i}, i=1,\dots,3(M-1)+1\}. \label{eqn:extension}
	\end{align}
	Set $\mathcal{T}^{(n)}$ contains polyblock corner points which yield a tighter upper bound for the boundary of (\ref{optc:d1r}) compared to $\mathcal{T}^{(n-1)}.$ Next, $\boldsymbol{y}^{(n)}$ is chosen as the point $\boldsymbol{y}$ in $\mathcal{T}^{(n)}$ fulfilling
	\begin{align}
		\boldsymbol{y}^{(n)} = \arg\max_{\boldsymbol{y}}\{R_{\eta}^{\prime\mathrm{D\uparrow\star}}(\boldsymbol{y}),\,\forall\,\boldsymbol{y} \in \mathcal{T}^{(n)}\}, \label{eqn:ymaxi}
	\end{align}
	where $R_{\eta}^{\prime\mathrm{D\uparrow\star}}(\boldsymbol{y})$ is the maximum of the convex optimization problem 
		\begin{maxi!}
		{p_{k,l}'^{\mathrm{D}} \geq 0\,\forall\,k,l}
		{R_{\eta}^{\prime\mathrm{D\uparrow}}}
		{\label{opt:d1ic}}
		{R_{\eta}^{\prime\mathrm{D\uparrow\star}}(\boldsymbol{y}) = }
		\addConstraint{g_0}{\leq y_0}{}
		\addConstraint{g_{s,l'}}{\leq y_{s,l'}}{\quad\forall\,s=1,2,3,l'=2,\dots,M,}
	\end{maxi!}
	with parameter $\boldsymbol{y}$ and optimization variables $p_{k,l}'^{\mathrm{D}},k=1,2,l=1,\dots,L,$ which can be obtained using standard convex optimization techniques \cite{Boyd2004}. Boundary point $\boldsymbol{z}^{(n)} = \mu^{(n)} \boldsymbol{y}^{(n)}$ is obtained via a line search, as in (\ref{eqn:linesearch}). Once the optimal value of (\ref{opt:d1ic}) with $\boldsymbol{z}^{(n)}$ and $\boldsymbol{y}^{(n)},$ i.e., $R_{\eta}^{\prime\mathrm{D\uparrow\star}}(\boldsymbol{z}^{(n)})$ and $R_{\eta}^{\prime\mathrm{D\uparrow\star}}(\boldsymbol{y}^{(n)}),$ have converged, upto a numerical tolerance $\epsilon,$ the iterations are stopped. Upon convergence, the solution to problem (\ref{opt:d1ic}) with $\boldsymbol{z}^{(n)},$ i.e., $R_{\eta}^{\prime\mathrm{D\uparrow\star}}(\boldsymbol{z}^{(n)})$, is a globally optimal solution of (\ref{opt:d1r}) \cite[Th. 1]{Tuy2000}. The algorithm is summarized in Algorithm  \ref{alg:g1}.
	
	\begin{figure}
		\begin{algorithm}[H]
			\small
			\begin{algorithmic}[1]
				\STATE {Initialize $\boldsymbol{y}^{(0)}$ as in (\ref{eqn:y0}), numerical tolerance $\epsilon,$ and iteration index $n=0.$}
				\STATE {Set $\mathcal{T}^{(0)} = \{\boldsymbol{y}^{(0)}\},$ and obtain $\boldsymbol{z}^{(0)}$ through line search as in (\ref{eqn:linesearch}).}
				\REPEAT	
				\STATE {$n \leftarrow n + 1$}
				\STATE {Compute $\mathcal{T}^{(n)} = \mathcal{T}^{(n-1)} \setminus \boldsymbol{y}^{(n-1)} \cup \{\boldsymbol{y}^{(n-1),i}, i=1,\dots,3(M-1)+1\}$ as in (\ref{eqn:extension}).}
				\STATE {Solve $\boldsymbol{y}^{(n)} = \arg\max_{\boldsymbol{y}}\{R_{\eta}^{\prime\mathrm{D\uparrow\star}}(\boldsymbol{y})\,\forall\,\boldsymbol{y} \in \mathcal{T}^{(n)}\}$ as in (\ref{eqn:ymaxi}).}
				\STATE {Obtain $\boldsymbol{z}^{(n)}$ through line search analogous to (\ref{eqn:linesearch}).}
				\UNTIL {$|R_{\eta}^{\prime\mathrm{D\uparrow\star}}(\boldsymbol{y}^{(n)}) - R_{\eta}^{\prime\mathrm{D\uparrow\star}}(\boldsymbol{z}^{(n)})| < \epsilon$}
				\STATE{Return solution $p_1'^{\mathrm{D}},p_2'^{\mathrm{D}},p_{k,l}'^{\mathrm{D}},k=1,2,l=1,\dots,L,$ to $R_{\eta}^{\prime\mathrm{D\uparrow\star}}(\boldsymbol{z}^{(n)}),$ obtained by solving (\ref{opt:d1ic}), as the power allocation.}
			\end{algorithmic}
			\caption{Globally optimal algorithm for solving (\ref{opt:d1r}).}
			\label{alg:g1}
		\end{algorithm}
	\end{figure}
	
	\begin{remark}
		In each iteration of the algorithm, (\ref{opt:d1ic}) is solved $3(M-1)+1$ times. Furthermore, the set $\mathcal{T}^{(n)}$ grows by $3(M-1)$ polyblock corner points in every iteration. In order to speed up convergence and to limit the size of $\mathcal{T}^{(n)},$ pruning techniques, such as removing improper polyblock vertices \cite[Sec. 3]{Tuy2000}, and restarting \cite[Secs. 4 and 5]{Tuy2000} can be utilized.
	\end{remark}

	\subsubsection{The Lower Bound}
	\label{sec:lb}
	Based on the lower bound given in (\ref{eqn:bounds}), the optimal solution of (\ref{opt:d1}$\downarrow$), $R_{\eta}^{\prime\mathrm{D\downarrow\star}},$ can be obtained by solving optimization problem
	\begin{maxi!}
		{p_{k,l}^{\mathrm{D}} \geq 0\,\forall\,k,l}
		{\sum_{l=1}^{L} \left[\eta R_{1,l}^{{\mathrm{D}}\downarrow} + (1-\eta) R_{2,l}^{{\mathrm{D}}}\right]}
		{\label{opt:d2}}
		{R_{\eta}^{\mathrm{D\downarrow\star}}=}
		\addConstraint{\sum_{l=1}^{L} (p_{1,l}^{\mathrm{D}}+ p_{2,l}^{\mathrm{D}})\label{optc:d2}}{\leq P_\mathrm{T}^{\mathrm{D}}}{},
	\end{maxi!}
	analogously to (\ref{opt:d1}$\uparrow$).

	Lastly, the upper and lower bounds for the outer boundary of the maximum rate region for downlink ST MIMO-NOMA are obtained by solving (\ref{opt:d1}$\uparrow$) and (\ref{opt:d1}$\downarrow$) for different $\eta \in [0,1],$ respectively, and taking the convex closure \cite{Rockafellar1970} of the obtained rate regions.
	\section{Simulation Results}
	\label{sec:sim}	
	In this section, we compare the ergodic achievable rate regions of the proposed uplink and downlink ST MIMO-NOMA schemes with those of existing MIMO-NOMA schemes and OMA. For both uplink and downlink transmission, we assume that the first and the second user are located at distances $d_1 = 250~\text{m}$ and $d_2 = 50~\text{m}$ from the BS, respectively. The path loss is modeled as $\Pi_k = d_k^2,$ i.e., $\Pi_1 = 250^2$ and $\Pi_2 = 50^2,$ and the noise variance is set as $\sigma^2 = -35~\text{dBm}.$ The elements of the channel matrices $\boldsymbol{H}_k \in \mathbb{C}^{M_k\times N}, k=1,2,$ are drawn from independent and identically distributed (i.i.d.) random variables $[\boldsymbol{H}_k]_{ij} \sim \mathcal{CN}(0,1),i=1,\dots,M_k,$ $j=1,\dots,N,$ $k=1,2.$ For uplink transmission, the maximum transmit powers of the users are set to $P_1^{\mathrm{U}} = 30 \text{ dBm},$ and $P_2^{\mathrm{U}} = 20 \text{ dBm},$ and for downlink transmission, a maximum transmit power of $P_\mathrm{T}^{\mathrm{D}} = 30 \text{ dBm}$ at the BS is adopted. The ergodic achievable rate regions of the considered schemes are computed by averaging the corresponding achievable rates over $10^4$ realizations of $\boldsymbol{H}_1$ and $\boldsymbol{H}_2,$ resulting in a 99\% confidence interval of $\pm 10^{-2}$ for the estimated ergodic achievable rates.

	\subsection{Uplink}
	\label{sec:simuplink}
	In the following figures, the ergodic achievable rate region of the proposed uplink (UL) ST MIMO-NOMA scheme is compared with those of ZF MIMO-NOMA \cite{Chen2016, Chen2016a, Higuchi2015}, UL GSVD MIMO-NOMA \cite{Ma2016}, SVD MIMO-NOMA, and OMA as well as the MIMO-MAC upper bound \cite{Yu2001}. The ergodic achievable rate region of the proposed scheme is obtained as described in Section \ref{sec:uopt}. Results for UL GSVD MIMO-NOMA \cite{Ma2016} are presented for the case $M_1=M_2=N.$ For SVD MIMO-NOMA, SVD-based precoding and detection matrices \cite{Gamal2011} are utilized for both users in order to diagonalize their MIMO channels. For decoding order (D-1-2), the symbols of the first user are decoded element-by-element treating the symbols of the second user as noise. Next, SIC is performed to eliminate the interference caused by the decoded symbols and the signal of the second user is subsequently diagonalized and decoded. For decoding order (D-2-1), a similar procedure is used. For OMA, time division multiple access (TDMA) with time fractions $\tau \in [0,1]$ and $(1-\tau)$ allocated to the first and the second user, respectively, is adopted. Furthermore, for OMA, the transmit powers of the users are normalized as $P_1'^{\mathrm{U}} = P_1^{\mathrm{U}}/\tau$ and $P_2'^{\mathrm{U}} = P_2^{\mathrm{U}}/(1-\tau)$ in order to obtain average powers $P_1^{\mathrm{U}}$ and $P_2^{\mathrm{U}}$ over the entire time slot.
	
	Figures \ref{fig:u446} and \ref{fig:u224} show the ergodic achievable rate regions for the case $M_1+M_2 > N$ (in particular, $M_1=M_2=4,N=6$), where ST MIMO-NOMA transmits $\bar{M}_1$ and $\bar{M}_2$ symbols inter-user-interference free by exploiting the null spaces of the MIMO channels of the users, and the case $M_1+M_2 = N$ (in particular, $M_1=M_2=2,N=4$), where the BS has sufficient DoFs to perform spatial orthogonalization, respectively. From the figures, we observe that, in both cases, the proposed ST MIMO-NOMA scheme significantly outperforms SVD and ZF MIMO-NOMA and OMA, and has a small gap to the MIMO-MAC upper bound. This small gap is expected because, in order to reduce decoding complexity, the proposed scheme cancels the received signal components that correspond to the off-diagonal elements of the triangularized channel matrices instead of exploiting them for decoding. The improved performance of ST MIMO-NOMA compared to SVD and ZF MIMO-NOMA is attributed to the fact that, in ST MIMO-NOMA, $\bar{M}_2$ symbols of user 2 and $\bar{M}_1$ symbols of user 1 experience no inter-user-interference.
	
	Next, in Figure \ref{fig:u444}, we consider the case $M_1=M_2=N=4,$ where both MIMO channel matrices have full rank. In this case, we observe that ST MIMO-NOMA exhibits a larger gap to the MIMO-MAC upper bound owing to the cancellation of the received signal components corresponding to the off-diagonal elements which, unlike the previous cases, cannot be partially compensated because the MIMO channel matrices do not have null spaces. Nevertheless, the proposed ST MIMO-NOMA scheme outperforms ZF MIMO-NOMA, UL GSVD MIMO-NOMA, and OMA for most user rates. SVD MIMO-NOMA has a marginally larger ergodic achievable rate region compared to ST MIMO-NOMA owing to the use of SVD-based precoding, which yields a better performance for user rates close to the single-user (SU)-MIMO rates. Nevertheless, for rates close to the SU-MIMO rates, the performance of the proposed ST MIMO-NOMA can be enhanced by utilizing a hybrid scheme that performs time sharing between OMA and the proposed ST MIMO-NOMA scheme.
	
	Lastly, Figure \ref{fig:u244} considers the asymmetric case $\bar{M}_1 = 0, \bar{M}_2 > 0$ (in particular, $M_1=2,M_2=4,N=6)$, where user 1 cannot benefit from the null space of the MIMO channel matrix of user 2. From the figure, we observe that for (D-2-1), the proposed ST MIMO-NOMA scheme has a small gap to the MIMO-MAC upper bound, whereas for (D-1-2), the gap for the proposed scheme is larger. This is because, for (D-2-1), $\bar{M}_2$ symbols of user 2, $s_{2,M+1},\dots,s_{2,M+\bar{M}_2},$ experience no inter-user-interference from user 1. On the other hand, for (D-1-2), all symbols of user 1, $s_{1,1},\dots,s_{1,M},$ experience degradation due to inter-user-interference from the symbols of user 2, as $\bar{M}_1=0.$

	\begin{figure*}
		\begin{minipage}[t]{0.48\textwidth}
			\centering
			\includegraphics[width=0.9\textwidth]{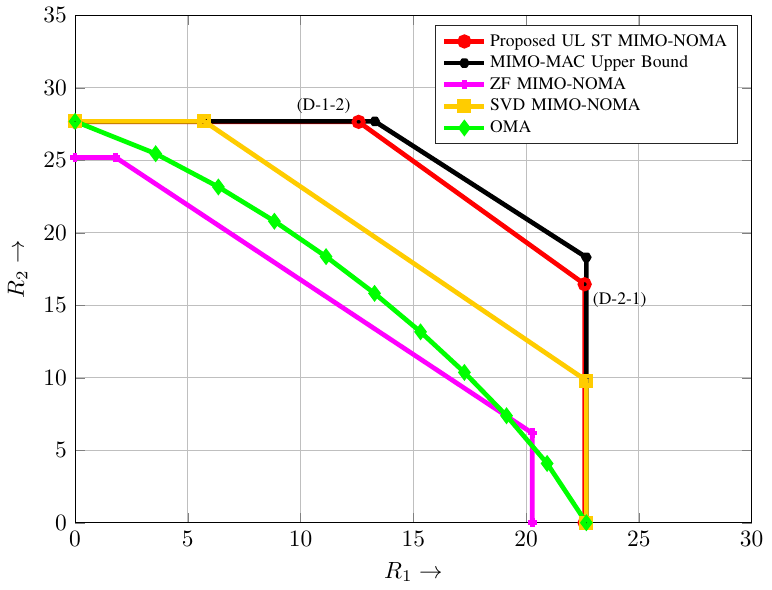}
			\caption{Uplink ergodic achievable rate region for $M_1,M_2=4, N=6, P_1^{\mathrm{U}} = 30 \text{ dBm},$ and $P_2^{\mathrm{U}} = 20 \text{ dBm}.$}
			\label{fig:u446}
	\end{minipage}%
	\hspace{0.04\textwidth}%
	\begin{minipage}[t]{0.48\textwidth}
			\centering
			\includegraphics[width=0.9\textwidth]{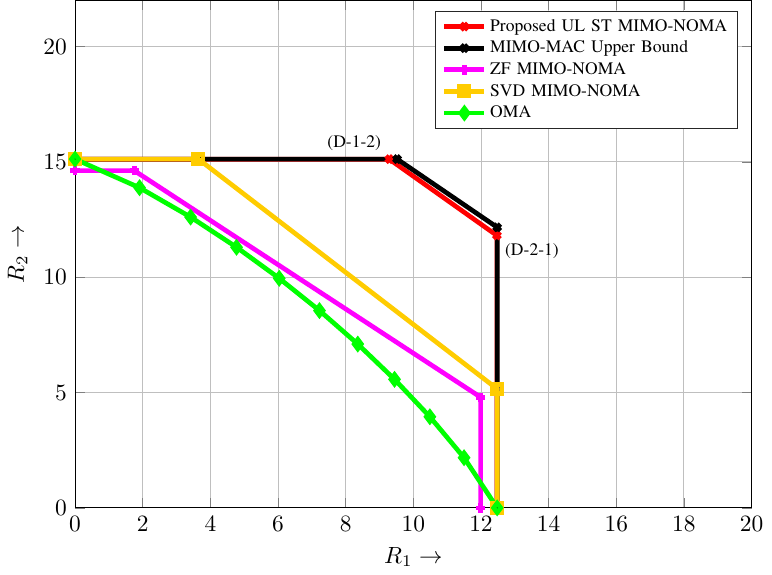}
			\caption{Uplink ergodic achievable rate region for $M_1,M_2=2, N=4, P_1^{\mathrm{U}} = 30 \text{ dBm},$ and $P_2^{\mathrm{U}} = 20 \text{ dBm}.$}
			\label{fig:u224}
			\vspace{0.5cm}
	\end{minipage}

	\begin{minipage}[t]{0.48\textwidth}
		\centering
		\includegraphics[width=0.9\textwidth]{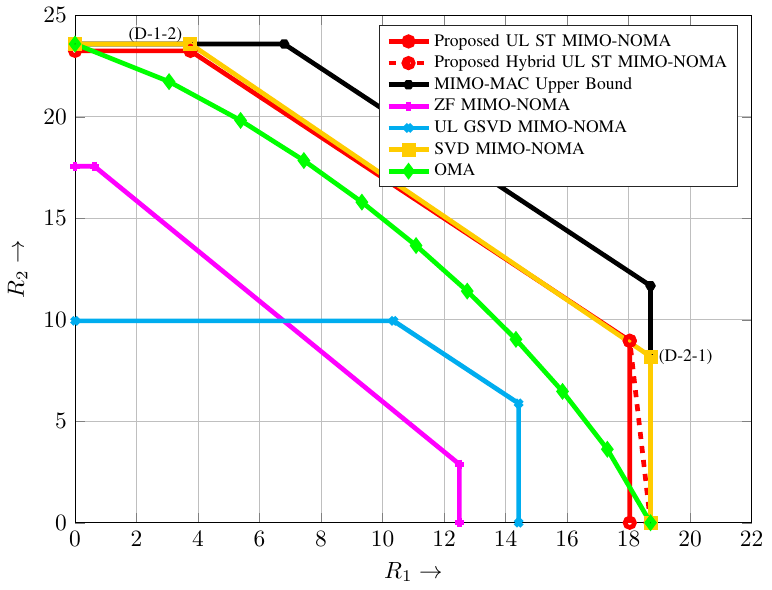}
		\caption{Uplink ergodic achievable rate region for $M_1,M_2,N=4, P_1^{\mathrm{U}} = 30 \text{ dBm},$ and $P_2^{\mathrm{U}} = 20 \text{ dBm}.$}
		\label{fig:u444}
	\end{minipage}%
	\hspace{0.04\textwidth}%
	\begin{minipage}[t]{0.48\textwidth}
		\centering
		\includegraphics[width=0.9\textwidth]{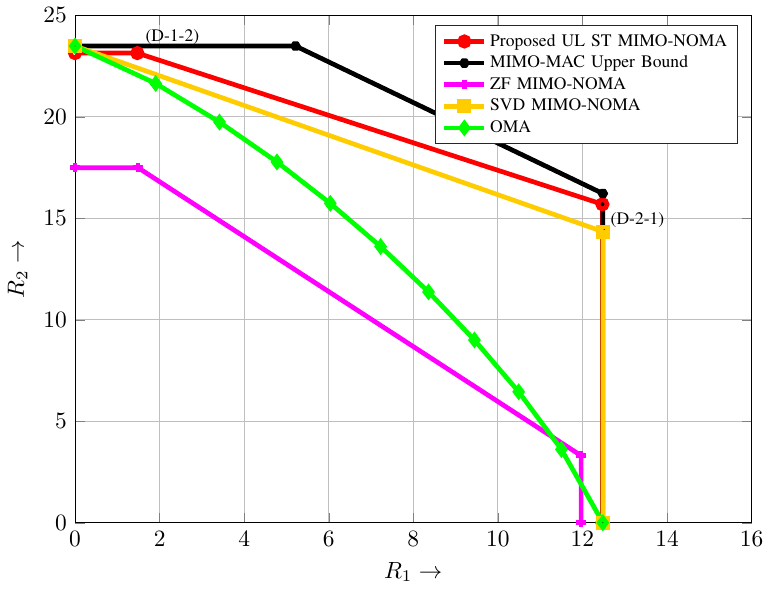}
		\caption{Uplink ergodic achievable rate region for $M_1=2,M_2,N=4, P_1^{\mathrm{U}} = 30 \text{ dBm},$ and $P_2^{\mathrm{U}} = 20 \text{ dBm}.$}
		\label{fig:u244}
		\vspace{0.5cm}
	\end{minipage}
	\end{figure*}
	
	\subsection{Downlink}
	\label{sec:dsim}
	For downlink (DL) transmission, we compare the ergodic achievable rate region of the proposed DL ST MIMO-NOMA scheme in Section \ref{sec:propdownlink} with those obtained for dirty paper coding (DPC), SD MIMO-NOMA  in \cite{Krishnamoorthy2020} and \cite{Chen2019}, and OMA. The ergodic achievable rate region of the proposed scheme is characterized via the upper and lower bounds (UB and LB) from Section \ref{sec:dopt}, the DPC upper bound is obtained by exploiting the MIMO BC-MAC duality \cite{Vishwanath2003}, and the ergodic achievable rate regions of the SD MIMO-NOMA schemes in \cite{Krishnamoorthy2020} and \cite{Chen2019}, which utilize GSVD and user-assisted simultaneous diagonalization (UA-SD) for simultaneously diagonalizing the MIMO channels of the users, are obtained via power allocation analogous to the proposed scheme but with a different BC-MAC transformation based on \cite{Jindal2004}. The ergodic achievable rate region for OMA is obtained by time sharing between the SU-MIMO rates. We first note that in Figures \ref{fig:d446}-\ref{fig:d244}, the upper and lower bounds of the proposed ST MIMO-NOMA scheme coincide, thereby providing an accurate characterization of the corresponding rate region.
	
	Figure \ref{fig:d446} shows the ergodic achievable rate region for the case $M_1+M_2 > N,$ with $M_1,M_2=4,$ and $N=6.$ We observe that the ergodic achievable rate region of the proposed scheme has a small gap of about 2 bits per channel use (BPCU) to the DPC upper bound. Moreover, we observe that the proposed scheme outperforms OMA for a wide range of user rates, owing to its ability to exploit the $\bar{M}_2=2$ and $\bar{M}_1=2$ null space dimensions of the MIMO channel matrices of users 1 and 2 for inter-user-interference free communication of $\bar{M}_1$ and $\bar{M}_2$ symbols of users 1 and 2, respectively. Furthermore, the proposed ST MIMO-NOMA outperforms SD MIMO-NOMA as it avoids channel inversion\footnote{SD MIMO-NOMA performs channel inversion which results in an increased transmit power, see \cite{Krishnamoorthy2020}.}, c.f. \cite{Krishnamoorthy2020} and \cite{Chen2019}, which leads to an enhanced performance. OMA is superior for rate pairs close to SU-MIMO. The gap to the DPC upper bound and the performance loss compared to OMA for user rates close to SU-MIMO are caused by the cancellation of the received signal components corresponding to the off-diagonal elements of the triangularized channel matrix. The energy of these signal components cannot be exploited for decoding. However, as for uplink transmission, performance can be further improved by utilizing a hybrid scheme, shown exemplarily for this case, which performs time sharing between OMA and the proposed downlink ST MIMO-NOMA scheme.
	
	Figure \ref{fig:d224} shows the ergodic achievable rate region for the case $M_1+M_2 = N,$ with $M_1,M_2=2,$ and $N=4.$ From the figure, we note that in this case the SD precoding scheme in \cite{Krishnamoorthy2020} yields a slightly larger rate region compared to the proposed ST precoding scheme. This is because, for $M_1+M_2=N,$ both the proposed ST precoder and the SD precoder in \cite{Krishnamoorthy2020} exploit the null space of the MIMO channel matrices of the users to achieve inter-user-interference free communication. However, in the proposed scheme, as explained earlier, the received signal components corresponding to the off-diagonal elements of the triangularized channel matrix are not exploited for decoding, whereas the SD MIMO-NOMA scheme in \cite{Krishnamoorthy2020} is able to exploit all signal components of the diagonalized channel matrix leading to a marginally larger rate region. However, we note that the computational complexity of the proposed ST MIMO-NOMA precoder, which exploits the QR decomposition, is lower than that of the SD MIMO-NOMA precoder in \cite{Krishnamoorthy2020}, cf. Section \ref{sec:dcomp}. Furthermore, the proposed scheme outperforms OMA and the SD MIMO-NOMA scheme in \cite{Chen2019} for most user rates. The hybrid scheme (not shown) can also be utilized in this case to enhance the ergodic achievable rate region. Lastly, for the considered system parameters, the proposed scheme has a relatively large gap to the DPC upper bound compared to the other considered scenarios. This is because, in order to achieve inter-user-interference free communication, the proposed downlink ST MIMO-NOMA scheme utilizes only the null spaces of the MIMO channels of the users. The column spaces of the MIMO channels of the users, $\mathrm{col}\big(\boldsymbol{H}_1^\mathrm{H}\big) \cup \mathrm{col}\big(\boldsymbol{H}_2^\mathrm{H}\big),$ are not exploited. On the other hand, uplink ST MIMO-NOMA, cf. Figure \ref{fig:u224}, exploits both the null and the column spaces of the MIMO channels of the users, resulting in a smaller gap to the corresponding upper bound.

	In Figures \ref{fig:d444} and \ref{fig:d244}, we consider the cases $M_1=M_2=N=4,$ where the MIMO channels of both users have full rank, and $M_1=2, M_2=N=4,$ where the MIMO channel of user 1 has a null space of dimension $\bar{M}_2=2,$ and the MIMO channel of user 2 has full rank, i.e., $\bar{M}_1=0.$ For both scenarios, the proposed downlink ST MIMO-NOMA outperforms SD MIMO-NOMA in \cite{Chen2019} and \cite{Krishnamoorthy2020} and OMA, and has a small gap to the DPC upper bound.
	
	\begin{figure*}
	\begin{minipage}[t]{0.48\textwidth}
		\centering
		\includegraphics[width=0.9\textwidth]{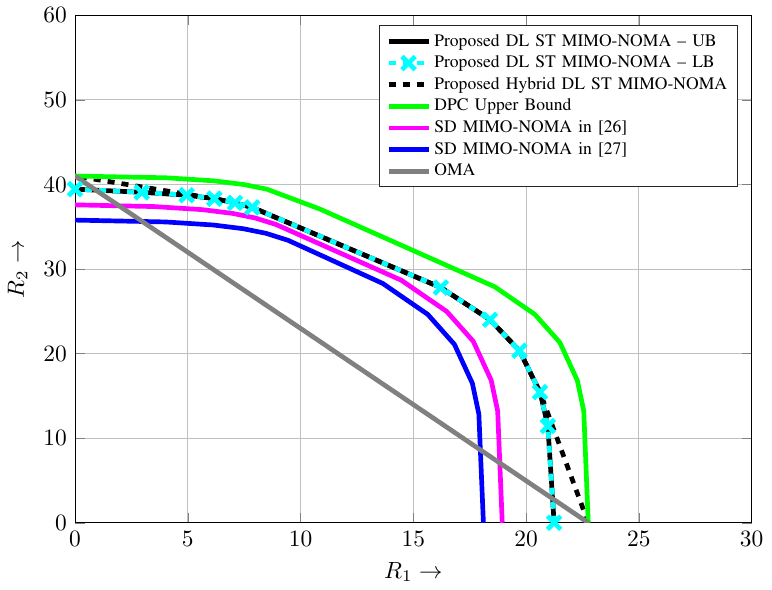}
		\caption{Downlink ergodic achievable rate region for $M_1,M_2=4,N=6,$ and $P_\mathrm{T}^{\mathrm{D}} = 30 \text{ dBm}.$}
		\label{fig:d446}
	\end{minipage}%
	\hspace{0.04\textwidth}%
	\begin{minipage}[t]{0.48\textwidth}
		\centering
		\includegraphics[width=0.9\textwidth]{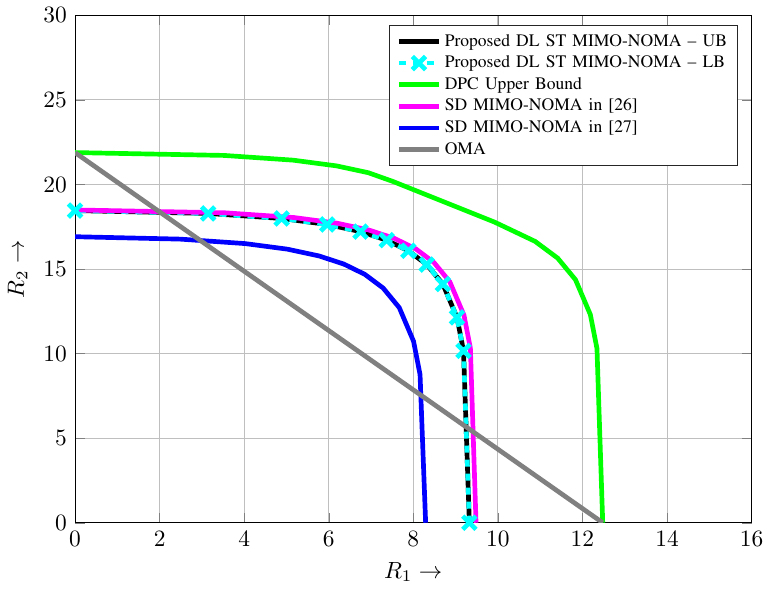}
		\caption{Downlink ergodic achievable rate region for $M_1,M_2=2,N=4,$ and $P_\mathrm{T}^{\mathrm{D}} = 30 \text{ dBm}.$}
		\label{fig:d224}
		\vspace{0.5cm}
	\end{minipage}
	
	\begin{minipage}[t]{0.48\textwidth}
		\centering
		\includegraphics[width=0.9\textwidth]{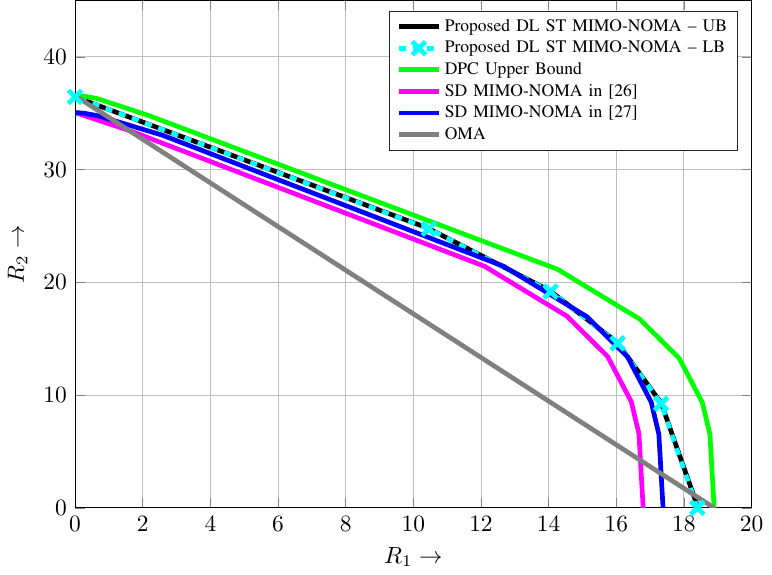}
		\caption{Downlink ergodic achievable rate region for $M_1,M_2,N=4,$ and $P_\mathrm{T}^{\mathrm{D}} = 30 \text{ dBm}.$}
		\label{fig:d444}
	\end{minipage}%
	\hspace{0.04\textwidth}%
	\begin{minipage}[t]{0.48\textwidth}
		\centering
		\includegraphics[width=0.9\textwidth]{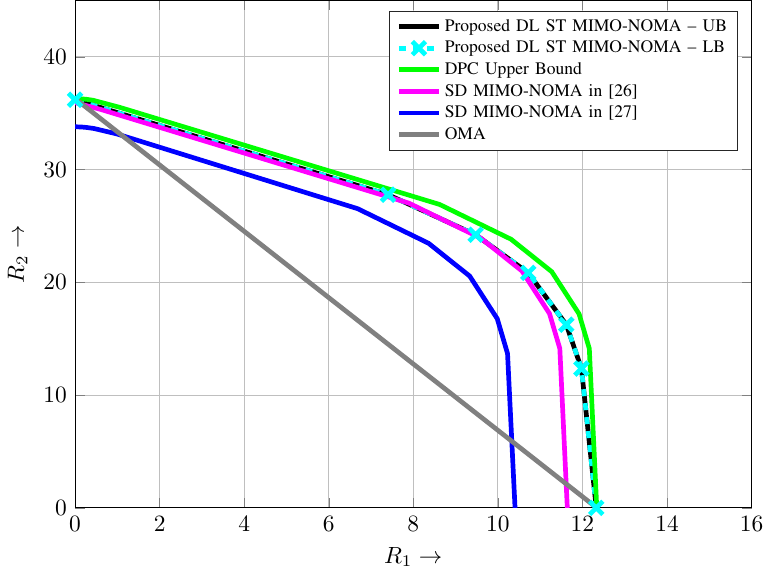}
		\caption{Downlink ergodic achievable rate region for $M_1=2,M_2,N=4,$ and $P_\mathrm{T}^{\mathrm{D}} = 30 \text{ dBm}.$}
		\label{fig:d244}
		\vspace{0.5cm}
	\end{minipage}
	\end{figure*}

	\section{Conclusion}
	\label{sec:con}
	We proposed novel uplink and downlink ST MIMO-NOMA precoding and decoding schemes that simultaneous triangularize the MIMO channel matrices of the users enabling low-complexity decoding, assuming self-interference cancellation at the receiver. The proposed uplink and downlink schemes exploit the null spaces of the MIMO channel matrices of the users to enable inter-user-interference free communication, and unlike SD MIMO-NOMA, avoid channel inversion at transmitter and receiver, which additionally enhances the ergodic rate performance. For uplink transmission, we characterized the maximum achievable rate region of the proposed ST MIMO-NOMA scheme utilizing convex optimization, and for downlink transmission, we exploited a BC-MAC transformation and an efficient polyblock outer approximation method. Computer simulations revealed that, for both uplink and downlink, the proposed ST MIMO-NOMA schemes perform close to the corresponding upper bounds and significantly outperform the considered baseline MIMO-NOMA schemes and OMA for most channel conditions and user rates. Further performance improvements were obtained with a hybrid scheme which performs time sharing between the proposed ST MIMO-NOMA schemes and SU-MIMO.
	
	The extension of the proposed ST MIMO-NOMA schemes to more than two users and their performance analysis for (a) practical modulation and coding schemes and (b) imperfect channel state information and imperfect SIC are interesting avenues for future research. Furthermore, the combination of ST MIMO-NOMA with PIC-based decoding is also a promising area for further study.
	\begin{appendices}
	\renewcommand{\thesection}{\Alph{section}}
	\renewcommand{\thesubsection}{\thesection.\arabic{subsection}}
	\renewcommand{\thesectiondis}[2]{\Alph{section}:}
	\renewcommand{\thesubsectiondis}{\thesection.\arabic{subsection}:}
	\section{Proofs}
	\label{app:proofs}
	\subsection{Proof of Theorem \ref{th:stu}}
	\label{app:stu}
	If $q_1 \geq p,$ let $\boldsymbol{A}_1^\mathrm{H} = \check{\boldsymbol{\mathcal{Q}}}_1\begin{bmatrix}
	\boldsymbol{0} \\ \check{\boldsymbol{R}}_1^\mathrm{H}
	\end{bmatrix}$ be the QL\footnote{The QL decomposition of a matrix, which decomposes the matrix into the product of a unitary matrix and a lower-triangular matrix, can be performed via the Gram-Schmidt procedure or Givens rotations \cite{Golub2012}.} decomposition of $\boldsymbol{A}_1^\mathrm{H},$ where $\check{\boldsymbol{\mathcal{Q}}}_1 \in \mathbb{C}^{q_1\times q_1}$ is a unitary matrix and $\check{\boldsymbol{R}}_1 \in \mathbb{C}^{p\times p}$ is an upper-triangular matrix with real-valued entries on the main diagonal. Let $\boldsymbol{V}_1$ contain the last $p$ columns of $\check{\boldsymbol{\mathcal{Q}}}_1$. Hence, $\boldsymbol{A}_1 \boldsymbol{V}_1 = \check{\boldsymbol{R}}_1.$ Furthermore, choose $\boldsymbol{U}=\boldsymbol{I}_p.$ Otherwise, if $q_1 < p,$ let $\boldsymbol{A}_1 = \check{\boldsymbol{\mathcal{Q}}}_1\begin{bmatrix}
	\check{\boldsymbol{R}}_1 \\ \boldsymbol{0}
	\end{bmatrix}$ be the QR decomposition of matrix $\boldsymbol{A}_1,$ where $\check{\boldsymbol{\mathcal{Q}}}_1 \in \mathbb{C}^{p\times p}$ is a unitary matrix and $\check{\boldsymbol{R}}_1  \in \mathbb{C}^{q_1\times q_1}$ is an upper-triangular matrix with real-valued entries on the main diagonal. Choose $\boldsymbol{U} = \check{\boldsymbol{\mathcal{Q}}}_1^\mathrm{H}$ and $\boldsymbol{V}_1 = \boldsymbol{I}_{q_1}.$
	
	Next, let via QL decomposition, $(\boldsymbol{U} \boldsymbol{A}_2)^\mathrm{H} = \check{\boldsymbol{\mathcal{Q}}}_2\begin{bmatrix}\boldsymbol{0} & \boldsymbol{0} \\ \check{\boldsymbol{X}}^\mathrm{H} & \check{\boldsymbol{R}}_2^\mathrm{H}\end{bmatrix},$ where $\check{\boldsymbol{\mathcal{Q}}}_2 \in \mathbb{C}^{q_2\times q_2}$ is a unitary matrix, $\check{\boldsymbol{R}}_2 \in \mathbb{C}^{\mathrm{min}\left\{p,q_2\right\}\times \mathrm{min}\left\{p,q_2\right\}}$ is an upper-triangular matrix with real-valued entries along the main diagonal, and $\check{\boldsymbol{X}} \in \mathbb{C}^{(p-\mathrm{min}\left\{p,q_2\right\})\times \mathrm{min}\left\{p,q_2\right\}}$ is a full matrix. Let $\boldsymbol{V}_2$ contain the last $\mathrm{min}\left\{p,q_2\right\}$ columns of $\check{\boldsymbol{\mathcal{Q}}}_2.$	Then, we have $\boldsymbol{U} \boldsymbol{A}_2 \boldsymbol{V}_2 = \begin{bmatrix} \check{\boldsymbol{X}}\\\check{\boldsymbol{R}}_2\end{bmatrix}.$ Using $\boldsymbol{U}, \boldsymbol{V}_1,$ and $\boldsymbol{V}_2$ from above, we obtain (\ref{eqn:stu1}) and (\ref{eqn:stu2}). \qed
	
	\subsection{Proof of Theorem \ref{th:std}}
	\label{app:std}
	Let  $\bar{\boldsymbol{H}}_1 \in \mathbb{C}^{N\times \bar{M}_1}$ and $\bar{\boldsymbol{H}}_2 \in \mathbb{C}^{N\times \bar{M}_2}$ be matrices that contain a basis for the null space of $\boldsymbol{H}_1$ and $\boldsymbol{H}_2,$ respectively. Let $\boldsymbol{K} \in \mathbb{C}^{N\times M}$ denote the matrix containing a basis for the null space of $\begin{bmatrix} \bar{\boldsymbol{H}}_1^\mathrm{H} & \bar{\boldsymbol{H}}_2^\mathrm{H}\end{bmatrix}.$ When the null spaces of $\boldsymbol{H}_1$ and $\boldsymbol{H}_2$ are trivial, i.e., when $M_1, M_2 \geq N,$ then $\boldsymbol{K} = \boldsymbol{I}_N.$ Let, by QR decomposition,
	\begin{align}
	\hat{\boldsymbol{\mathcal{Q}}}_1 \boldsymbol{R}_1^{\mathrm{D}} &= \boldsymbol{H}_1 \begin{bmatrix}\boldsymbol{K} & \bar{\boldsymbol{H}}_2\end{bmatrix}, \label{eqn:qr1}\\
	\hat{\boldsymbol{\mathcal{Q}}}_2 \boldsymbol{R}_2^{\mathrm{D}} &= \boldsymbol{H}_2 \begin{bmatrix}\boldsymbol{K} & \bar{\boldsymbol{H}}_1\end{bmatrix}. \label{eqn:qr2}
	\end{align}
	Then, (\ref{eqn:std1}) and (\ref{eqn:std2}) are satisfied by setting
	\begin{align}
	\boldsymbol{X}^{\mathrm{D}} = \begin{bmatrix}\boldsymbol{K} & \bar{\boldsymbol{H}}_2 & \bar{\boldsymbol{H}}_1\end{bmatrix}, \label{eqn:x}
	\end{align}
	and choosing $\boldsymbol{Q}_1^{\mathrm{D}} = \hat{\boldsymbol{\mathcal{Q}}}_1^\mathrm{H}$ and $\boldsymbol{Q}_2^{\mathrm{D}} = \hat{\boldsymbol{\mathcal{Q}}}_2^\mathrm{H}$ from (\ref{eqn:qr1}) and (\ref{eqn:qr2}) above, to obtain
	\begin{align}
	\scalebox{0.9}{\mbox{\ensuremath{\displaystyle \boldsymbol{Q}_1^{\mathrm{D}}\boldsymbol{H}_1\boldsymbol{X}^{\mathrm{D}} = \begin{bmatrix}\underbrace{\hat{\boldsymbol{\mathcal{Q}}}_1^\mathrm{H} \boldsymbol{H}_1 \begin{bmatrix}\boldsymbol{K} & \bar{\boldsymbol{H}}_2\end{bmatrix}}_{\boldsymbol{R}_1^{\mathrm{D}}} &  \underbrace{\hat{\boldsymbol{\mathcal{Q}}}_1^\mathrm{H} \boldsymbol{H}_1 \bar{\boldsymbol{H}}_1}_{\boldsymbol{0}}\end{bmatrix}, \quad
	\boldsymbol{Q}_2^{\mathrm{D}}\boldsymbol{H}_2\boldsymbol{X}^{\mathrm{D}} \overset{(a)}{=} \begin{bmatrix} \underbrace{\hat{\boldsymbol{\mathcal{Q}}}_2^\mathrm{H} \boldsymbol{H}_2 \boldsymbol{K}}_{\boldsymbol{R}_2'^{\mathrm{D}}} & \underbrace{\hat{\boldsymbol{\mathcal{Q}}}_2^\mathrm{H} \boldsymbol{H}_2 \bar{\boldsymbol{H}}_2}_{\boldsymbol{0}} &  \underbrace{\hat{\boldsymbol{\mathcal{Q}}}_2^\mathrm{H} \boldsymbol{H}_2 \bar{\boldsymbol{H}}_1}_{\boldsymbol{R}_2''^{\mathrm{D}}}\end{bmatrix}}}},
	\end{align}
	where (a) holds because the QR decomposition in (\ref{eqn:qr2}) is unaffected by the zero columns introduced in the middle. \qed
	
	\subsection{Proof of Theorem \ref{th:trans}}
	\label{app:trans}
	By substituting the expressions in (\ref{eqn:p1t}) and (\ref{eqn:p2t}) into (\ref{eqn:bc}), we obtain
	\begin{align}
	&\scalebox{0.85}{\mbox{\ensuremath{\displaystyle \sum_{l=1}^{M}\left[\eta \log_2\left(1 + \frac{1}{\Pi_1 \sigma^2} \left|\varrho^{(1)}_{l,l}\right|^2 p_{1,l}'^{\mathrm{D}}\right) + (1-\eta) \log_2\left(1 + \frac{\frac{1}{\Pi_2} \left|\varrho^{(2)}_{l,l}\right|^2 p_{2,l}'^{\mathrm{D}}}{\sigma^2 + \frac{1}{\Pi_1} \left|\varrho^{(1)}_{l,l}\right|^2 p_{1,l}'^{\mathrm{D}}}\right)\right]}}} \nonumber\\
	&\quad \scalebox{0.85}{\mbox{\ensuremath{\displaystyle {}+  \eta\sum_{l=M+1}^{M+\bar{M}_1}\log_2\left(1 + \frac{1}{\Pi_1\sigma^2}p_{1,l}'^{\mathrm{D}}\left|\varrho^{(1)}_{l,l}\right|^2\right) + (1-\eta) \sum_{l=M+\bar{M}_1+1}^{M+\bar{M}_1+\bar{M}_2} \log_2\left(1 + \frac{1}{\Pi_2\sigma^2}p_{2,l}'^{\mathrm{D}}\left|\varrho^{(2)}_{l-\bar{M}_1,l-\bar{M}_1}\right|^2\right)}}},
	\end{align}
	which is the weighted sum rate of $L$ parallel SISO-MAC channels under a sum power constraint and can be simplified to (\ref{eqn:mac1}) and (\ref{eqn:mac2}) based on \cite[Theorem 1]{Liu2008}.
	
	Next, for the sum power constraint in (\ref{eqn:pe}), from (\ref{eqn:p1t}) and (\ref{eqn:p2t}), we have
	\begin{align}
	\scalebox{0.85}{\mbox{\ensuremath{\displaystyle \sum_{l=1}^{M} (p_{1,l}^{\mathrm{D}} + p_{2,l}^{\mathrm{D}})}}} &= \scalebox{0.95}{\mbox{\ensuremath{\displaystyle \sum_{l=1}^{M} \Big[\underbrace{\frac{p_{1,l}'^{\mathrm{D}}}{\sigma^2} (\sigma^2 + \frac{1}{\Pi_1}p_{2,l}^{\mathrm{D}}\left|\varrho^{(1)}_{l,l}\right|^2) + \frac{p_{2,l}'^{\mathrm{D}}\sigma^2}{\sigma^2 + \frac{1}{\Pi_1}p_{1,l}'^{\mathrm{D}}\left|\varrho^{(1)}_{l,l}\right|^2 }}_{\overset{(a)}{=} p_{1,l}'^{\mathrm{D}}+p_{2,l}'^{\mathrm{D}}} {} + p_{1,l}'^{\mathrm{D}}\frac{1}{\Pi_1\sigma^2}\sum_{l'=l+1}^{M}p_{2,l'}^{\mathrm{D}}\left|\varrho^{(1)}_{l,l'}\right|^2\Big]}}} \nonumber\\
	&= \scalebox{0.85}{\mbox{\ensuremath{\displaystyle \sum_{l=1}^{M} \Big[p_{1,l}'^{\mathrm{D}}+p_{2,l}'^{\mathrm{D}} + p_{1,l}'^{\mathrm{D}}\frac{1}{\Pi_1}\sum_{l'=l+1}^{M}\frac{p_{2,l'}'^{\mathrm{D}}}{\sigma^2 + \frac{1}{\Pi_1}p_{1,l'}'^{\mathrm{D}}\left|\varrho^{(1)}_{l',l'}\right|^2 }\left|\varrho^{(1)}_{l,l'}\right|^2\Big]}}} \nonumber\\
	&= \scalebox{0.85}{\mbox{\ensuremath{\displaystyle \sum_{l=1}^{M} (p_{1,l}'^{\mathrm{D}} + p_{2,l}'^{\mathrm{D}}) + \sum_{l'=2}^{M}\frac{p_{2,l'}'^{\mathrm{D}}}{\sigma^2 + \frac{1}{\Pi_1}p_{1,l'}'^{\mathrm{D}}\left|\varrho^{(1)}_{l',l'}\right|^2 }\sum_{l=1}^{l'-1}\frac{1}{\Pi_1} p_{1,l}'^{\mathrm{D}} \left|\varrho^{(1)}_{l,l'}\right|^2}}}.
	\end{align}
	where (a) is obtained by substituting $p_{2,l}^{\mathrm{D}}$ from (\ref{eqn:p2t}) and simplifying the resulting expression. Lastly, based on (\ref{eqn:p1t}) and (\ref{eqn:p2t}), as $p_{k,l}'^{\mathrm{D}} = p_{k,l}^{\mathrm{D}},k=M+1,\dots,L,$ (\ref{eqn:pe}) follows. \qed
	\end{appendices}
	
	\bibliographystyle{IEEEtran}
	\bibliography{references}
\end{document}